\documentclass[review,onefignum,onetabnum]{siamart190516}

\usepackage[english]{babel}
\usepackage{amssymb}
\usepackage{enumerate,enumitem}
\usepackage{cite,etoolbox}
\usepackage{multirow,bigdelim,blkarray}

\usepackage[caption=false]{subfig}
\graphicspath{{Figures/}}
\hypersetup{unicode=true}

\nolinenumbers

\crefname{con}{Condition}{Conditions}
\crefalias{enumi}{con}
\newcommand{\change}[1]{#1}
\makeatletter 
\pretocmd\@bibitem{\color{black}\csname keycolor#1\endcsname}{}{\fail}
\newcommand\citedcolor[1]{\@namedef{keycolor#1}{\color{blue}}}
\makeatother

\usepackage{lipsum}
\usepackage{amsfonts}
\usepackage{graphicx}
\usepackage{epstopdf}

\ifpdf
  \DeclareGraphicsExtensions{.eps,.pdf,.png,.jpg}
\else
  \DeclareGraphicsExtensions{.eps}
\fi

\newsiamremark{remark}{Remark}
\newsiamremark{hypothesis}{Hypothesis}
\crefname{hypothesis}{Hypothesis}{Hypotheses}
\newsiamthm{claim}{Claim}

\usepackage{amsopn}

\newcommand{\bs}{\boldsymbol}
\newcommand{\bb}{\mathbb}
\newcommand{\mcal}{\mathcal}

\newcommand{\eye}{\bs{I}}
\newcommand{\zero}{\bs{0}}
\newcommand{\one}{\bs{1}}

\newcommand{\lb}{\left(}
\newcommand{\rb}{\right)}
\newcommand{\ls}{\left[}
\newcommand{\rs}{\right]}
\newcommand{\lc}{\left\{}
\newcommand{\rc}{\right\}}
\newcommand{\ld}{\left.}

\newcommand{\lv}{\left\vert}
\newcommand{\rv}{\right\vert}
\newcommand{\lV}{\left\Vert}
\newcommand{\rV}{\right\Vert}

\newcommand{\lcl}{\left\lceil}
\newcommand{\rcl}{\right\rceil}
\newcommand{\LRV}[1]{{\left\vert\kern-0.25ex\left\vert\kern-0.25ex\left\vert #1 \right\vert\kern-0.25ex\right\vert\kern-0.25ex\right\vert}}
\newcommand{\T}{\mathsf{T}}

\newcommand{\nat}{\bb{N}}

\newcommand{\nth}{^\mathsf{th}}

\newcommand{\rank}[1]{\mathsf{Rank}\lc#1\rc}

\newcommand{\supp}[1]{\mathsf{Supp}\lc#1\rc}

\newcommand{\range}[1]{\mathcal{CS}\lc#1\rc}

\newcommand{\matA}{\bs{A}}

\newcommand{\matD}{\bs{D}}

\newcommand{\matW}{\bs{W}}

\newcommand{\bbC}{\bb{C}}

\newcommand{\bbP}{\bb{P}}

\newcommand{\bbR}{\bb{R}}
\newcommand{\bbS}{\bb{S}}

\newcommand{\bbZ}{\bb{Z}}

\newcommand{\calB}{\mcal{B}}
\newcommand{\calC}{\mcal{C}}

\newcommand{\calE}{\mcal{E}}
\newcommand{\calF}{\mcal{F}}

\newcommand{\calH}{\mcal{H}}
\newcommand{\calI}{\mcal{I}}

\newcommand{\calL}{\mcal{L}}
\newcommand{\calM}{\mcal{M}}

\newcommand{\calP}{\mcal{P}}

\newcommand{\calS}{\mcal{S}}
\newcommand{\calT}{\mcal{T}}
\newcommand{\calU}{\mcal{U}}

\newcommand{\calW}{\mcal{W}}
\newcommand{\calX}{\mcal{X}}
\newcommand{\calY}{\mcal{Y}}

\newcommand{\veca}{\bs{a}}

\newcommand{\vecu}{\bs{u}}
\newcommand{\vecv}{\bs{v}}
\newcommand{\vecw}{\bs{w}}
\newcommand{\vecx}{\bs{x}}
\newcommand{\vecy}{\bs{y}}
\newcommand{\vecz}{\bs{z}}

\newcommand{\vecalpha}{\bs{\alpha}}

\newcommand{\matLambda}{\bs{\Lambda}}

\newcommand{\matPhi}{\bs{\Phi}}
\newcommand{\matPsi}{\bs{\Psi}}

\newcommand{\incomp}{\calT_{\mathrm{incomp}}}
\newcommand{\comp}{\calT_{\mathrm{comp}}}
\newcommand{\CS}[1]{\calC\calS\lc#1\rc}
\newcommand{\lar}{\calT_{\mathrm{large}}}
\newcommand{\smal}{\calT_{\mathrm{small}}}
\newcommand{\ER}{Erd\H os-R\' enyi }
\newcommand{\UER}{\mathcal{ER} }
\newcommand{\DER}{\mathcal{DER} }
\newcommand{\Ber}{\mathrm{Ber} }
\begin{document}

\title{Controllability of Network Opinion in Erd\H os-R\' enyi Graphs Using Sparse Control Inputs\thanks{\funding{This work was supported in part by National Science Foundation under grant ENG 60064237 and 1714180, and the U. S. Army Research Office under grant  W911NF-19-1-0365.}}}

\author{Geethu Joseph\thanks{Department of Electrical Engineering and Computer Science, Syracuse University, Syracuse, NY~13244, USA
  (\email{gjoseph@syr.edu},\email{varshney@syr.edu}).}
\and Buddhika Nettasinghe\thanks{School of Electrical and Computer Engineering, Cornell University, Ithaca, NY~14853, USA
  (\email{dwn26@cornell.edu}, \email{vikramk@cornell.edu}).}
\and Vikram Krishnamurthy\footnotemark[3]
\and Pramod K. Varshney\footnotemark[2]}

\maketitle
\begin{abstract}
This paper considers a social network modeled as an \ER random graph. Each individual in the network updates her opinion using the weighted average of the opinions of her neighbors. We explore how an external manipulative agent can drive the opinions of these individuals to a desired state with a limited additive influence on their innate opinions. 
 We show that the manipulative agent can steer the network opinion to any arbitrary value in finite time (i.e., the system is controllable) almost surely when there is no restriction on her influence. However, when the control input is sparsity constrained, the network opinion is  controllable with some probability. We lower bound this probability using the concentration properties of random vectors based on the L\' evy concentration function and small ball probabilities. Further, through numerical simulations, we  compare the probability of controllability in \ER graphs with that of power-law graphs to illustrate the key differences between the two models in terms of controllability.  Our theoretical and numerical results shed light on how controllability of the network opinion depends on the parameters such as the size and the connectivity of the network, and the sparsity constraints faced by the manipulative agent. 
\end{abstract}
\begin{keywords}
	Social network opinion, linear propagation, sparse-controllability, \ER graph, concentration inequalities
	\end{keywords}
\section{Introduction}
\label{sec:intro}


Consider the following problem regarding controllability of network opinion: the opinions propagated over an \ER random graph consisting of $N$ people are influenced by an external agent (referred to as the \textit{manipulative agent} henceforth)~\cite{ghaderi2013opinion}. The goal of the manipulative agent is to influence the opinions of the people so that the network opinion is driven to a desired state. However, the manipulative agent is subject to sparsity constraints: it can only influence a few people in the network at each time.

Our formulation is as follows: Let $k=1,2,\ldots$ denote discrete time. The network opinion at time $k$, denoted by $\vecx_k\in\bbR^N$ follows a DeGroot type linear propagation model~\cite{degroot1974reaching}:
\begin{equation}
\label{eq:system_model}
\vecx_{k}=\bar{\matA}\vecx_{k-1}+\vecu_k,
\end{equation}
 where  $\bar{\matA}\in\bbR^{N\times N}$ is the row-normalized weighted adjacency matrix of the network.  In the most general setting, the randomness of the system $\eqref{eq:system_model}$ is induced by a hierarchical probability measure over the space of graphs with $N$ nodes where the locations of non-zero entries of $\bar{\matA}$ are modeled using an \ER graph and the non-zero entries of each row of $\bar{\matA}$ are drawn from a continuous distribution on the unit simplex. 
The control input from the manipulative agent, $\vecu_k\in\bbR^{N}$ represents her influence at time $k$. We assume that the manipulative agent can influence only one of the predefined (overlapping) groups of people in the network at any time instant, and the size of each such group is small compared to the network size ${N}$. Thus, $\vecu_k$ is $s-$sparse with $s\ll N$ (i.e.,~\textit{budget-constrained}), and its support belongs to a set of admissible support sets $\calU$ (i.e., \textit{pattern-constrained}). 
 
A fundamental question that arises in this context is: \emph{Can the network opinion be driven to a desired value using the budget and pattern constrained inputs within a finite time duration?} The answer is that it is possible with some probability (that arises due to the randomness in \ER model). We use concentration inequalities to derive a lower bound for this probability which is a function of the network size $N$, the sparsity $s$, the edge probability $p$ and the set of admissible support sets $\calU$ for two versions of the \ER graph model: a directed graph and an undirected graph. The main results of the paper are informally stated  below:

\begin{theorem}[Informal statement of main results]
\label{thm:informal}
Consider the linear opinion formation model in \eqref{eq:system_model} where $\bar{\matA}$ is  the row-normalized adjacency matrix of an \ER graph. Assume that the control input $\vecu_{k}$ is $s$-sparse with $0<s\ll N$ and its support~(defined in \Cref{tab:notation}) satisfies  $\supp{\vecu_k}\in\calU$ for some set $\calU$ which depends on the specific sparsity pattern (detailed in \Cref{sec:sparsity}). Further, assume that the  nonzero entries of each row $\bar{\matA}$ is drawn from a continuous distribution on the unit simplex. Then, the following hold:
\begin{enumerate}[label=(\alph*),leftmargin=0.5cm]
\item \emph{Undirected graph:} If the edge probability $p$ satisfies
$( N-s)^{-1} \leq p \leq  1-( N-s)^{-1}$,
then the probability of controllability of the network opinion in \eqref{eq:system_model} is at least $o\lb\bar{Q}(\calU)(1-p)^{sN}\ls1-
e^{-c(p(N-s))}\rs\rb$. 
\item \emph{Directed graph:} If the edge probability $p$ satisfies
$C \frac{\log(N-s)}{N-s} < p \leq  1-C \frac{\log(N-s)}{N-s}$,
then the probability of controllability of the network opinion in \eqref{eq:system_model} is at least  $o\lb\bar{Q}(\calU)(1-p)^{sN}\ls1
-e^{-c(p(N-s))}\rs\rb$, 
\end{enumerate}
where $\bar{Q}(\calU)\geq1$ is an increasing function of $s$ and $N$ which depends on the set $\calU$. Here $C,c>0$ are universal constants.
\end{theorem}

The key insights gained from \Cref{thm:informal} are as follows: 
\begin{itemize}[leftmargin=0.3cm]
\item The probability of controllability increases with the network size $N$ and sparsity $s$ when all other parameters are kept constant. In particular, the network opinion can be controlled almost surely, if the network size is sufficiently large ($N\to\infty$), irrespective of the sparsity patterns. 

\item The probability bound on controllability of network opinion is small when the network is either loosely connected (small values of $p$ which are close to 0) or densely connected (large values of $p$ which are close to 1), i.e.,~the probability of the system being controllable is larger for moderate values of $p$.

\item The dependence of controllability on the sparsity pattern is captured by the function $\bar{Q}$. This relation allows us to compare the probability of controllability under various popular \change{sparsity structures}  like unconstrained, piece-wise, and block sparsity.
\end{itemize}
%
%
%
The technique that we use to prove our results relies on  rank-related properties of binary random matrices (that model the adjacency matrix of \ER graph). The key tools used to derive these properties are the L\' evy concentration function and small ball probability.  The analysis characterizes the smallest 
singular value of the (unweighted) binary adjacency matrix of an \ER graph, which can be of independent interest. 

\subsection{Practical context of the model and examples} \label{sec:practical}
Controllability of the network opinion has applications in marketing~\cite{nabi2014social}, targeted fake-news campaigns~\cite{shu2019studying}, and political advertising~\cite{cremonini2017controllability}. For such problems, the randomness of the \ER model captures the unknown structure of the underlying social network. The directed graphs represent social networks such as Twitter whereas the undirected graphs represent social networks such as Facebook. 

The DeGroot type opinion propagation model (given in \eqref{eq:system_model}) that we consider is an average based process that has been used widely in the literature~(\cite{golub2012homophily, banerjee2019naive, jackson2010social, wai2016active}) to model consensus and learning in multi-agent systems. Here, each individual in the network averages the opinions of her neighbors according to the non-uniform weights\footnote{The model in \cite{golub2012homophily, banerjee2019naive} uniformly averages the opinions of the neighbors whereas we consider the more general case of averaging based on social reputation. } prescribed by the $i\nth$ row of $\bar{\matA}$.

The sparsity pattern models the limited influence of  the manipulative agent in the following examples. 
  
\begin{enumerate}[leftmargin=0.5cm]
	\item Consider a company that sends one salesperson each to $m$ different parts of a country for marketing their products by offering free samples. At each time instant, the company can afford a total of $s$ free samples (and hence, constraining the sparsity of the control input to be~$s$), and each salesperson can visit at most $s/m$ potential customers (and thus, constraining the sparsity pattern). Hence, the control input $\vecu_k$ at each time instant $k$ is formed by concatenating $m$ sparse vectors. This is an example of the sparsity pattern called \emph{piece-wise sparsity} (\Cref{def:piece}). 
	
	\item Consider a candidate visiting different parts of a country during election as part of a political campaign. At each time instant, she can visit only a particular area and influence the group of people located there. Therefore, the input $\vecu_k$ at each time instant $k$ has nonzero entries occurring in clusters (corresponding to the individuals located in a particular area). This is an example of the sparsity pattern called \emph{block sparsity} or \emph{group sparsity} (\Cref{def:block}).
\end{enumerate}

\subsection{Related work} 
The linear opinion propagation model  in \eqref{eq:system_model} is similar to the widely used DeGroot model~\cite{degroot1974reaching}. For such linear models, several works have explored the problem of controlling opinions in a social network by drawing tools from control theory.  In \cite{rahmani2009controllability}, the authors  explore how a set of individuals (called the leaders or strategic agents) can be used to manipulate the opinions of a fixed set of people who are in their neighborhood. The leaders cannot strategically choose the people whom they can influence, but can design the control inputs suitably. On the contrary, our model is more flexible as the manipulative agent is an external entity for the network, and she can influence the individuals of her choice. Other similar works in \cite{dietrich2016opinion, liu2014control} also follow the model of strategic agents. Another study presented in \cite{dhal2015vulnerability} deals with the network vulnerability where an external agent aims to disrupt the synchronized state (consensus) of a networked system. Unlike our model which limits the number of people influenced by the manipulative agent, the model in  \cite{dhal2015vulnerability} constrains the total energy spent by the agent. This work \cite{dhal2015vulnerability} provides guarantees related to the inverse of controllability Gramian and its associated statistics. Further, a recent study in \cite{wendt2019control} models the control input from the manipulating agent as a local feedback control using a projection of the current network control and presents guarantees on Hurwitz stability of the equivalent feedback system. Moreover, in \cite{zhao2016bounded,zhao2018understanding}, the authors examine how the nodes of a social network can be classified as opinion leaders and opinion followers during the opinion formation process amidst factors such as trustworthiness and uncertainty in decision making. A detailed survey of such models can be found in \cite{aydougdu2017interaction}. Furthermore, the work in \cite{pasqualetti2014controllability} is closely related to ours  where the authors investigate the difficulty of controlling the opinion dynamics with metrics based on the control energy and explore the trade-off between the control energy and the number of control nodes. However, none of the existing works explore controllability of a networked system generated by a specific random graph model under the constraints on the sparsity pattern of the control input (which is the aim of this paper). A key reason for this literature gap is that the Kalman-rank based test for sparse-controllability~(i.e., controllability of a linear dynamical system under sparsity constraints on the inputs) is combinatorial which makes the analysis cumbersome. Recently, a simpler rank-based test which is similar to the classical PBH test~\cite{hautus1970stabilization,kailath1980linear} and equivalent to the Kalman type rank test has been presented in \cite{joseph2020controllability}. Therefore, we build upon the controllability conditions provided in \cite{joseph2020controllability} to derive sufficient conditions for sparse-controllability of network opinion. 

\vspace{0.2cm}
\change{\noindent \textbf{Organization:} \Cref{sec:systemmodel} presents the network opinion dynamics model. \Cref{sec:generalized} introduces the notion of generalized sparse-controllability and derives  the necessary and sufficient conditions for a deterministic linear dynamical system to be sparse-controllable. Based on this result, \Cref{sec:undirected} and \Cref{sec:directed} present the main results of our paper: the bounds on the probability with which the opinion dynamics for an \ER graph (directed or undirected) is sparse-controllable. Finally, \Cref{sec:simulations} presents numerical illustrations that complement and verify the main results.}

\vspace{0.2cm}
\noindent
\textbf{Notation:} Boldface lowercase letters denote vectors, boldface uppercase letters denote matrices, and calligraphic letters denote sets. \Cref{tab:notation} summarizes the other notations used throughout the paper. 
\begin{table}[ht]\caption{Summary of Notation}
\begin{tabular}{r l r l}
\hline\\
$\bbS^{N-1}$ &: Unit Euclidean sphere in $\bbR^N$\\
$\matA_i$ &: $i\nth$ column of $\matA$ &
$\matA_{ij}$ &: the $(i,j)\nth$ entry of $\matA$\\
$\odot$ &: Hadamard product &$\calP(a)$ &: Power set of $\lc1,2,\ldots,a\rc$\\
$\matA_{\calS}$ &\multicolumn{3}{l}{: Submatrix of $\matA$ formed by the columns  indexed by $\calS$}\\
 $\matA_{\calS,:}$ &\multicolumn{3}{l}{: Submatrix of $\matA$ formed by the  rows indexed by $\calS$}\\
 $\supp{\cdot}$ &\multicolumn{3}{l}{: Support of a vector, $\supp{\vecz}=\{i:\vecz_i\neq 0\}$ for any  $\vecz\in\bbR^N$}\\
$\mathrm{dist}(\veca,\calS)$ &\multicolumn{3}{l}{: $\ell_2$ distance of $\veca$ from $\calS$, $\mathrm{dist}(\veca,\calS)\triangleq\underset{\veca'\in\calS}{\inf}\lV\veca-\veca'\rV$}\\
\hline
\end{tabular}
\label{tab:notation}
\end{table}

\section{Opinion Dynamics Model}\label{sec:systemmodel}
\change{We consider a social network with $N$ individuals. The network opinion vector $\vecx_{k}\in\bbR^N$ (with element $\vecx_{k}[i]$ denoting the opinion of the $i\nth$ individual at time $k$) evolves according to \eqref{eq:system_model}, i.e., $\vecx_{k}=\bar{\matA}\vecx_{k-1}+\vecu_k$. Thus, the network opinion dynamics model has two key components: the random graph model that represents the probability distribution from which $\bar{\matA}$ is obtained;  and the sparsity model that imposes restrictions on the additive control inputs $\vecu_k$. We present them below:
}

\subsection{Random graph model}
\label{sec:connectedness}

\change{This subsection describes the probabilistic model for the matrix $\bar{A}$ in \eqref{eq:system_model} in terms of the underlying social network graph.. The underlying network is modeled as a weighted \ER graph~\cite{erdos1959on, bollobas2001random} which can be either an undirected graph (for networks like Facebook) or a directed graph (for networks like Twitter). The formal definitions of the undirected and directed \ER models are as follows:
	\begin{definition}[\ER model] 
		\label{def:ER}
		Let $\matW=\vecw\vecw^\T\in\bbR_+^{N\times N}$ where $\vecw$ is sampled from an arbitrary continuous distribution on $\bbR_+^N$. Also, let $\odot$ denote the Hadamard product of two matrices, and $\Ber(p)$ denote a Bernoulli distribution with parameter $p$. The row-normalized adjacency matrix
		\begin{equation}\label{eq:barA_defn}
		\bar{\matA}=\matLambda(\matA\odot\matW),
		\end{equation}
		with $\matLambda\in\bbR^{N\times N}$ being the normalizing diagonal matrix ($\matLambda_{ii}=\sum_{j=1}^N \lb \matA\odot\matW\rb_{ij}$), follows:
		\begin{enumerate}
			\item the undirected \ER distribution $\UER(N,\matW,p)$, if for $i,j=1,2,\ldots,N$,
			\begin{equation}
			\matA_{ij} \begin{cases}
			\overset{\text{iid}} {\sim} \Ber(p) & \text{ for } j<i\\
			= 0 & \text{ for } j=i
			\end{cases}   \text{ and } \matA_{ij}=\matA_{ji} \hspace{0.4cm}\text{ for } j>i; \label{eq:Erdos_Renyi_undirected}
			\end{equation}
			
			\item 
			the directed \ER  distribution  $\DER(N,\matW,p)$, if for $i,j=1,2,\ldots,N$,
			\begin{equation} \label{eq:Erdos_Renyi_directed}
			\matA_{ij} \begin{cases}
			\overset{\text{iid}} {\sim} \Ber(p) & \text{ for } j\neq i\\
			= 0 & \text{ for } j=i.
			\end{cases} 
			\end{equation}
		\end{enumerate}
	\end{definition}

A few words about \Cref{def:ER}. The matrix $\matA \in\lc0,1\rc^{N\times N}$ specifies the presence or absence of the edges in the graph, $\matW$ assigns positive weights to each existing edge, and $\matLambda$ ensures that the resulting matrix $\bar{\matA}$ is a stochastic matrix. The choice of $\matW$ as a rank 1 matrix (outer product $\vecw\vecw^\T$) makes our model more general than the unweighted \ER model (i.e., $\matW=\one\one^\T$) and useful in practical scenarios that require weighted edges. For example,  the weights can represent the reputation of an individual as discussed in \Cref{sec:practical}. Further, since $\matLambda$ normalizes each row of $\matA\odot\matW$ in \eqref{eq:barA_defn}, the entries of $\matW$ only determine the ratio of weights (i.e., $\matW=\vecw\vecw^\T$ and $\matW=\one\vecw^\T$ correspond to the same system). 

We note that $\bar{\matA}_{ij}$ specifies the trust that the $i\nth$ individual in the network has on her neighbor $j$ (i.e.,~$j$ such that $\matA_{ij} = 1$). Therefore, the first term in \eqref{eq:system_model} models how the neighbors affect the opinion formation. 
}

%
%
%

\subsection{Sparsity model}\label{sec:sparsity}
\change{We now discuss the second term $\vecu_k$ in  \eqref{eq:system_model} which models the influence from the manipulative agent.  The sparsity model refers to the constraints faced by the manipulative agent in the form of restrictions imposed on the support of the control inputs. 
We adopt a generalized sparsity model called the \emph{pattern-and-budget-constraint sparsity~(PBCS) model}.  The PBCS model is characterized by a set $\calU \subseteq\calP(N)$ called the \emph{admissible supports set} such that $\lv\calS\rv= s$ for each $\calS\in\calU$~(from \Cref{tab:notation}, $\calP(N)$ denotes the power set of $\lc 1,2,\ldots,N\rc$). Therefore,
	\begin{equation}
		\label{eq:unconstrained_sparsity_set}
		\calU\subseteq \calU_1\triangleq\lc\calS\subset\{1,2,\ldots,N\}:\lv\calS\rv= s\rc.
	\end{equation}	
	Then, any \emph{admissible control input} $\vecu_k$ chosen by the manipulative agent at any time instant~$k$ satisfies $\supp{\vecu_k}\subseteq \calS$ for some $\calS\in\calU$. Thus, an admissible control input is a vector whose support is a subset of an element in the admissible supports set $\calU$.
Three useful sparsity models that serve as examples of the PBCS model~\change{\cite{foucart2013mathematical,li2016piecewise}} are defined below:


\begin{definition}[Unconstrained sparsity model]\label{def:uncon}
A	PBCS model \eqref{eq:unconstrained_sparsity_set} is called an {unconstrained sparsity model} if the admissible supports set $\calU=\calU_1$ given in ~\eqref{eq:unconstrained_sparsity_set}.
\end{definition}
Thus, the unconstrained sparsity model~(\Cref{def:uncon}) is the least restrictive version of the PBCS model because from \eqref{eq:unconstrained_sparsity_set}, the admissible supports set $\calU_1$ has the largest cardinality among all the versions of the PBCS model.
\begin{definition}[Piece-wise sparsity model]\label{def:piece}
A PBCS model \eqref{eq:unconstrained_sparsity_set} is called a piece-wise sparsity model if the admissible supports set $\calU$ is given by
	\begin{equation*}
		\calU = \calU_2\triangleq \change{ \lc\calS: \calS = \cup_{i=1}^m\calS_i: \calS_i\subseteq \lc (i-1)N/m+j, 1\leq j \leq N/m\rc \text{ and } \lv\calS_i\rv=s/m\rc}.
	\end{equation*}
\end{definition}
The piece-wise sparsity model corresponds to the case where each admissible control input can be expressed as a concatenation of $m$ sparse vectors, each with sparsity of at most $s/m$ (assume that both $N$ and $s$ are divisible by $m$).
\begin{definition}[Block sparsity model]\label{def:block}
	A PBCS model \eqref{eq:unconstrained_sparsity_set} is called a block sparsity model if the admissible supports set $\calU$ is given by
	\begin{equation*}
		\calU = \calU_3\triangleq\change{\lc \calS :\calS= \cup_{i\in\calI}\calS_i: \calS_i\subseteq \{(i-1)m+j, 1\leq j\leq m\} \text{ and } \lv\calI\rv=s/m\rc}.
	\end{equation*}
\end{definition}
The block sparsity model corresponds to the case where the non-zero entries of each admissible control input form clusters of equal size $m\leq N$ (assume that both $N$ and $s$ are divisible by $m$). To give more insights, a brief discussion on the relationship between the block sparse vectors and the piece-wise sparse vectors is presented in \Cref{app:relation}.

Having specified the opinion dynamics model, we next derive the conditions under which the network opinion is controllable under the different sparsity models (defined in \Cref{sec:sparsity}) on graphs sampled from the \ER model given in \Cref{sec:connectedness}.  We start with a generalized sparse-controllability test in the next section. 
}
\section{Generalized Sparse-Controllability Results}\label{sec:generalized}
\change{This section presents necessary and sufficient conditions for sparse-controllability of a general linear dynamical system using control inputs from the PBCS model. In subsequent sections, we explore the probability with which these necessary and sufficient conditions are satisfied by the opinion dynamics model in~\eqref{eq:system_model}. Thus, the main result of this section serves as the starting point to derive the results in \Cref{sec:undirected,sec:directed}.}


\change{As the main result of this section~(necessary and sufficient conditions for sparse controllability of a linear dynamical system) is of interest to the broader field of dynamical systems, we express the result for a more general version of a linear dynamical system using notation that is different from Sec.~\ref{sec:systemmodel} and \eqref{eq:system_model}. More specifically, we consider the following general linear dynamical system:} 
\begin{equation}\label{eq:general_lds}
\vecalpha_k = \matPhi\vecalpha_{k-1}+\matPsi\vecv_{k},\hspace{1cm}k=1,2,\ldots,
\end{equation}
where $\vecalpha_k\in\bbR^N$ denotes the state vector, $\vecv_k\in\bbR^L$ denotes the control input, $\matPhi\in\bbR^{N\times N}$ denotes the state transition matrix, and $\matPsi\in\bbR^{N\times L}$ denotes the input matrix. 
For the system in \eqref{eq:general_lds}, the generalized sparse-controllability result is presented below. In the theorem below, we use $\vecv_k$ to represent a sparse input (see \Cref{sec:sparsity}).
\begin{theorem}[Controllability under PBCS model]
\label{thm:controlpattern}
Consider the linear system \eqref{eq:general_lds} with sparse input vectors $\vecv_k$. Let the admissible supports set be $\calU\subseteq\calP(L)$ with $\lv\calS\rv= s, \forall\calS\in\calU$. Then, for any initial state $\vecalpha_0$ and final state $\vecalpha_K$, there exists a sparse input sequence $\vecv_k$,  which steers the system from the state $\vecalpha_0$ to $\vecalpha_K$ for some finite $K$,  if and only if the following two conditions hold:
\begin{enumerate}[label={(\alph*)},leftmargin=0.5cm]
\item For all $\lambda\in\bbC$, the rank of $\begin{bmatrix}
\lambda\eye-\matPhi & \matPsi_{\calM}
\end{bmatrix}\in\bbR^{N\times (N+\lv\calM\rv)}$ is $N$, where the set $\calM=\cup_{\calS\in\calU}\calS\subseteq \{1,2,\ldots,L\}$. \label{con:control_pattern_a}
\item There exists a set $\calS\in\calU$ such that the rank of  $\begin{bmatrix}
\matPhi & \matPsi_{\calS}
\end{bmatrix} \in\bbR^{N\times (N+\lv\calS\rv)}$ is~$N$.
\label{con:control_pattern_b}
\end{enumerate}  
\end{theorem}
\begin{proof}
See \Cref{app:control_pattern}.
\end{proof}

Both conditions of \Cref{thm:controlpattern} have to be satisfied simultaneously for the system to be controllable under the sparsity constraints on the input. Condition \ref{con:control_pattern_a} is identical to the classical PBH test for controllability of the reduced  linear system described by the matrix pair $(\matPhi,\matPsi_{\calM})$. Here, $\matPsi_{\calM}$ is the effective control matrix, as the entries of the control inputs corresponding to the complement $\calM^c$ of $\calM$ are always zero. Since all sparse-controllable systems are controllable, necessity of the first condition is straight-forward. Condition \ref{con:control_pattern_b} provides the extra condition to be satisfied by a controllable system (i.e., satisfies condition \ref{con:control_pattern_a}) to be controllable under the sparsity constraints. 

We note that \Cref{thm:controlpattern}  is an extension of the sparse-controllability result~\cite[Theorem 1]{joseph2020controllability} for a more general PBCS model. The original result~\cite[Theorem 1]{joseph2020controllability} is applicable only to the unconstrained sparsity model in \Cref{def:uncon}. 

Next, we apply \Cref{thm:controlpattern} to our stochastic setting and derive probabilistic results on sparse-controllability of network opinion dynamics. Before we launch into those main results, we provide a result on controllability of the network opinion using unconstrained (non-sparse) inputs. \change{
\begin{proposition}\label{prop:control}
	Consider the generalized linear system \eqref{eq:general_lds} and assume $\vecalpha_k=\vecx_k$, $\vecv_k=\vecu_k$, $\matPhi=\bar{\matA}$, and $\matPsi=\eye$ (i.e.,~the systems in  \eqref{eq:general_lds} and \eqref{eq:system_model} are equivalent). If the sparsity $s = N$ (i.e., the manipulative agent can influence any number of people), then the system \eqref{eq:general_lds} (and hence, the system \eqref{eq:system_model}) is controllable. 
\end{proposition}
\begin{proof}
	The result follows immediately by substituting $\vecalpha_k=\vecx_k$, $\vecv_k=\vecu_k$, $\matPhi=\bar{\matA}$, and $\matPsi=\eye$ in the two conditions of Theorem~\ref{thm:controlpattern}.
\end{proof}
\Cref{prop:control} indicates that the opinion dynamics model in \eqref{eq:system_model} is always controllable when an arbitrary number of people can be influenced by the manipulative agent. Therefore, the non-trivial problem in this context is the controllability analysis in the sparse regime~(i.e.,~$s\ll N$) that we deal with in the subsequent sections. 
}

\section{Sparse Controllability of Opinions in an Undirected Graph}\label{sec:undirected}
\change{This section studies the probability with which the opinion dynamics system  \eqref{eq:system_model} satisfies the sparse-controllability conditions specified in \Cref{thm:controlpattern} when the underlying graph is sampled from the undirected \ER model $\UER(N, \matW, p)$ defined in \Cref{def:ER}. The main result of this section is a lower bound on the probability of controllability that is followed by a discussion of the insights that it yields. }

To state the result, we define the function $Q:\{0,1,\ldots,s\}\times \calP(N)\to\nat$ 
as follows: 
\begin{equation}
Q(t,\calU) \triangleq \lv \lc \calI\subseteq \calS: \calS\in\calU \text{ and } \lv\calI\rv=t\rc\rv.\label{eq:Q_defn}
\end{equation}
where $\nat$ is the set of natural numbers.
We recall that $\calU$ is the admissible supports set of the control input. Also, if $\calU^{(1)}\subset\calU^{(2)}$, we obtain that  $Q(t,\calU^{(1)})< Q(t,\calU^{(2)})$. 
Therefore, $Q(t,\calU)$ that counts the number of $t$-sized subsets of $\calS\in\calU$ and it can be considered as a measure of the flexibility of the sparsity pattern. If the sparsity model is less pattern-constrained, the size of $\calU$ increases which in turn increases $Q(t,\cdot)$. Intuitively, if the control input is less constrained, then the system is more likely to be controllable~\change{(as \Cref{prop:control} also suggests)}. This dependence of the probability of controllability on $\calU$ is captured by the function $Q$ as given by the following result. 

\change{The main result of this section is a lower bound (in terms of $Q$ in \eqref{eq:Q_defn} and the parameters of the undirected \ER model in \eqref{eq:Erdos_Renyi_undirected}) on the probability that the network opinion is controllable. In essence, \Cref{thm:undirected} suggests that the probability that the opinion dynamics model in \Cref{eq:system_model} is sparse controllable increases linearly  with the function $Q(\cdot,\calU)$ and exponentially with the network size $N$.}
\begin{theorem}[Sparse controllability of network opinion in undirected \ER graphs]\label{thm:undirected}
	\change{Consider the network opinion model  \eqref{eq:system_model}  where the constraints on control input $\vecu_k$ at each time $k$ are modeled by the PBCS model in \eqref{eq:unconstrained_sparsity_set} with the admissible supports set $\calU$ and the additional constraint $\lv\cup_{\calS\in\calU}\calS\rv=N$. Let the weighted adjacency matrix $\bar{\matA}$ be sampled from the undirected \ER model $\UER(N,\matW, p)$ given in \eqref{eq:Erdos_Renyi_undirected}.} Assume that
	\begin{equation}
	\label{eq:condition_unidrectedER}
	( N-s)^{-1} \leq p \leq  1-( N-s)^{-1}.
	\end{equation}
	Then, the network opinion of the system can be steered to any desired value from any initial network opinion in finite time, with probability at least $q$ where
	\begin{equation}\label{eq:undirected}
	q= \sum_{i=0}^s Q(i,\calU)(1-p)^{i(2N-i-1)/2} \ls1-C\exp\lb-c(p(N-i))^{1/32}\rb\rs,
	\end{equation}
	for some constants $C,c>0$, and $Q$ is as defined in \eqref{eq:Q_defn}.
\end{theorem}

\begin{proof}
See \Cref{app:undirected}.
\end{proof} 


We note that in \eqref{eq:undirected}, the exponent $1/32$ is not the sharpest possible power of $pN$ and it can be improved. The crux of the result is that there exists a small constant $0<\tilde{c}<1$ such that the probability of controllability exceeds  $1-C\exp\lb-c(pN)^{\tilde{c}}\rb$. Also, comparing \Cref{thm:undirected} and \Cref{thm:informal}, we note that $\bar{Q}(\calU)=\sum_{i=0}^s Q(i,\calU)$ in the informal statement of the main results in \Cref{sec:intro}.
The other implications of \Cref{thm:undirected} that highlight its importance in practical contexts are discussed below. 
\subsection{Dependence on parameters}\label{sec:dependence}
\change{In this subsection, we explore the effect of the parameters of the opinion dynamics model~\eqref{eq:system_model} such as the sparsity~$s$, edge probability~$p$, network size $N$ on the lower bound $q$ given in \Cref{thm:undirected}.}
\begin{itemize}[leftmargin=0.3cm]
\item \emph{Sparsity $s$:} The dependence of $q$ on the \change{sparsity} $s$ is only through the first summation term in \eqref{eq:undirected} and it increases with the sparsity $s$. This observation is intuitive because larger sparsity $s$ implies less restrictions on the control inputs, and thus, it leads to a higher probability of controllability.

\item \emph{Edge probability $p$: } \change{The sufficient condition \eqref{eq:condition_unidrectedER} for sparse controllability of opinion dynamics includes a range of values of the edge probability~$p$ that depends on both the network size $N$ and sparsity $s$}. This suggests that when the network is highly connected~\change{(i.e.~$p \approx 1$)} or sparsely connected~\change{(i.e., $p \approx 0$)}, it is difficult to \change{control} the \change{network opinion} by influencing a few number of people. This is intuitive as it is not possible to influence the opinion of the network if \change{$p\approx0$} as the people in the network do not influence each other \change{significantly due to lack of connections among them}. Also, when $p$ is close to 1, we see that $\bar{\matA}$ \change{has approximately low rank} with high probability, and hence the network opinion cannot be driven to all the vectors in the null space of $\bar{\matA}$. This is because from \eqref{eq:system_model},
\begin{equation}\label{eq:sparse_recovery}
\vecx_k = \bar{\matA}^k\vecx_0+\sum_{i=1}^{k-1}\bar{\matA}^{k-i}\vecu_i+\vecu_k,
\end{equation}
where the first two terms belong to the column space of $\bar{\matA}$. Therefore, the sparsity of $\vecu_k$ should be greater than the nullity of $\bar{\matA}$ to ensure controllability.

\item \emph{Network size $N$:} As the network size $N$ increases, the \change{lower} bound on the probability of controllability $q$ (in \eqref{eq:undirected}) increases. This is due to the following:
\begin{equation*}
q\geq  Q(0,\calU)\ls1-C\exp\lb-c(p(N))^{1/32}\rb\rs=1-C\exp\lb-c(p(N))^{1/32}\rb,
\end{equation*}
since from \eqref{eq:Q_defn}, $Q(0,\calU)=1$ for any set $\calU$. \change{This observation does not immediately imply that the probability of controllability increases with $N$. However,} as $N\to\infty$, and $s=o(N^{b})$, for some $b\in(0,1)$, the value $q$ goes to unity. This implies that the network opinion is controllable, almost surely, for any admissible supports set~$\calU$. This result is intuitive because the \ER model ensures that the expected number of neighbors of every person on the network is $p(N-1)$. Thus, as the network size increases, the people who can be influenced by the manipulative people, are connected to more people in the network. Consequently, the opinion of an asymptotically large network is  controllable with probability one.

\item \emph{Giant connected component:} If $p < 1/N$, the undirected \ER graph almost surely has no connected component of size $O(\log(n))$ whereas if $p > 1/N$, the undirected \ER graph has a unique giant connected component containing a positive fraction of the nodes almost surely~\cite{bollobas1984evolution,jackson2010social}. We note that the lower bound in \eqref{eq:condition_unidrectedER} is larger than this threshold value ($1/N$) of $p$ required for the almost sure existence of a unique giant component in the graph, and thus, our results hold only when a unique giant connected component exists almost surely. \change{Here, we do not make explicit assumptions on the connected components, and the almost sure existence of a unique giant component follow automatically from~\eqref{eq:condition_unidrectedER}. }

\item \emph{\change{Versions of the PBCS model}:}
\change{
We note that all piece-wise and block sparse vectors~(\Cref{def:piece,def:block}) belong to the set of unconstrained sparse vectors~(\Cref{def:uncon}). Thus,  if the network opinion is controllable using piece-wise sparse or block sparse control inputs, it is controllable with unconstrained sparse control inputs. Also, the block sparsity can be seen as a special case of piece-wise sparsity with a common support for all blocks (see \Cref{app:relation}). Therefore, we have\footnote{
Using straightforward computations, we get $Q(i,\calU_1)=\binom{N}{i}, $
\begin{equation*}
 Q(i,\calU_2) = \sum_{\substack{0\leq j_1, \ldots, j_m\leq \frac{s}{m}\\j_1+j_2\ldots+j_m=i}}\prod_{l=1}^m\binom{\frac{N}{m}}{j_l} \text{ and } Q(i,\calU_3) \!=\!
\sum_{l=\lcl \frac{im}{N} \rcl}^{\min\lc i,\frac{sm}{N}\rc}\binom{m}{l}\lb\frac{N}{m}\rb^l\binom{l\lb\frac{N}{m}-1\rb}{i-l}.
\end{equation*}
}
\begin{equation}\label{eq:sparsity_pattern_order}
Q(i,\calU_3)\leq Q(i,\calU_2)\leq Q(i,\calU_1)
\end{equation}
where $\calU_1,\calU_2,\calU_3\subseteq \calP(N)$ correspond to the three different versions of the PBCS model given in \Cref{def:uncon,def:piece,def:block}.
Thus, out of the three models, the unconstrained sparse vectors offer the highest probability of controlling network opinion, followed by the piece-wise sparse vectors.
}

\end{itemize}

\subsection{Design of sparse inputs}
\Cref{thm:rankundirected} deals with the existence of a set of sparse vectors such that the network opinion can be driven from any initial state $\vecx_0\in\bbR^N$ to any final state $\vecx_f\in\bbR^N$. The next important question is the design of this set of sparse vectors. The problem can be cast as a sparse vector recovery problem using \eqref{eq:sparse_recovery} where we solve for $\lc\vecu_i\rc_{i=1}^k$~\cite{sefati2015linear,sriram2020control,kafashan2016relating}. To elaborate, \cite[Theorem 3]{joseph2020controllability} ensures that we need at most $N$ control input vectors (i.e., $k=N$) to drive the system from any arbitrary initial state to the desired state. Therefore, we consider the following sparse vector recovery problem:
$\vecx_f-\matA^N\vecx_0 = \begin{bmatrix}
\matA^{N-1} & \matA^{N-2} & \ldots \eye
\end{bmatrix}\tilde{\vecu}$,
where $\tilde{\vecu} = \begin{bmatrix}
\vecu_1^\T & \vecu_2^\T &\hdots &\vecu_N^\T
\end{bmatrix}^T\in\bbR^{N^2}$ is a piece-wise sparse vector with $N$ blocks and each block being $s-$sparse. Then, the control input $\tilde{\vecu}$ can be estimated using the piece-wise sparse recovery algorithms~\cite{li2016piecewise,sriram2020control}. 

\section{Sparse Controllability of Opinions in  a Directed Graph}\label{sec:directed}
\change{In this section, we extend \Cref{thm:undirected} in \Cref{sec:undirected} to the directed \ER model $\DER(N,\matW,p)$ specified in \Cref{def:ER}. The main result of this section (\Cref{thm:directed}) establishes a lower bound on the probability of sparse-controllability of network opinion in a directed \ER graph. This probability bound increases linearly with the counting function $Q(\cdot,\calU)$ in \eqref{eq:Q_defn} (which quantifies the complexity of the sparsity model $\calU$) and exponentially with the network size~$N$. 
}
\begin{theorem}[Sparse controllability of network opinion in directed \ER graphs]\label{thm:directed}
	\change{Consider the network opinion model  \eqref{eq:system_model}  where the constraints on control input $\vecu_k$ at each time $k$ are modeled by the PBCS model in \eqref{eq:unconstrained_sparsity_set} with the admissible supports set $\calU$ and the additional constraint $\lv\cup_{\calS\in\calU}\calS\rv=N$. Let the weighted adjacency matrix $\bar{\matA}$ be sampled from the directed \ER model $\DER(N,\matW,p)$ given in \eqref{eq:Erdos_Renyi_directed}.} Assume that
	\begin{equation}\label{eq:pcon_directed}
	C (N-s)^{-1}\log(N-s) < p \leq  1-C (N-s)^{-1}\log(N-s).
	\end{equation}
	Then, the network opinion of the system can be steered to any desired value from any initial network opinion in finite time, with probability at least $q$ where 
	\begin{equation}\label{eq:directed}
	q= \sum_{i=0}^s Q(i,\calU)(1-p)^{i(N-1)}\ls1-\exp\lb-c(p(N-i))\rb\rs,
	\end{equation}
	for some constants $C,c>0$, and $Q$ is as defined in \eqref{eq:Q_defn}.
\end{theorem}

\begin{proof}
See \Cref{app:directed}.
\end{proof}
We make similar observations regarding the result as those discussed in \Cref{sec:undirected} about \Cref{thm:undirected}. Also, we note that the range of edge connectedness $p$ is shorter for directed graphs compared to undirected graphs. Here, we require the graph to be connected with higher probability which is $\log (N-s)$ times more than the undirected graph. For the same edge probability, the undirected graph is likely to have more number of connections and hence, it is more controllable.
Further, the bound on the probability of controllability is of similar order for directed graphs and undirected graphs, when the other parameters are kept the same. Therefore, the direction of information flow (uni-directional in directed graphs vs bi-directional in undirected graphs) does not have a significant effect on the probability bound.


\section{Numerical Experiments}\label{sec:simulations}
\begin{figure}
\begin{center}
\subfloat[Undirected graph with $N=12$.]
{
\includegraphics[width=5.7cm]{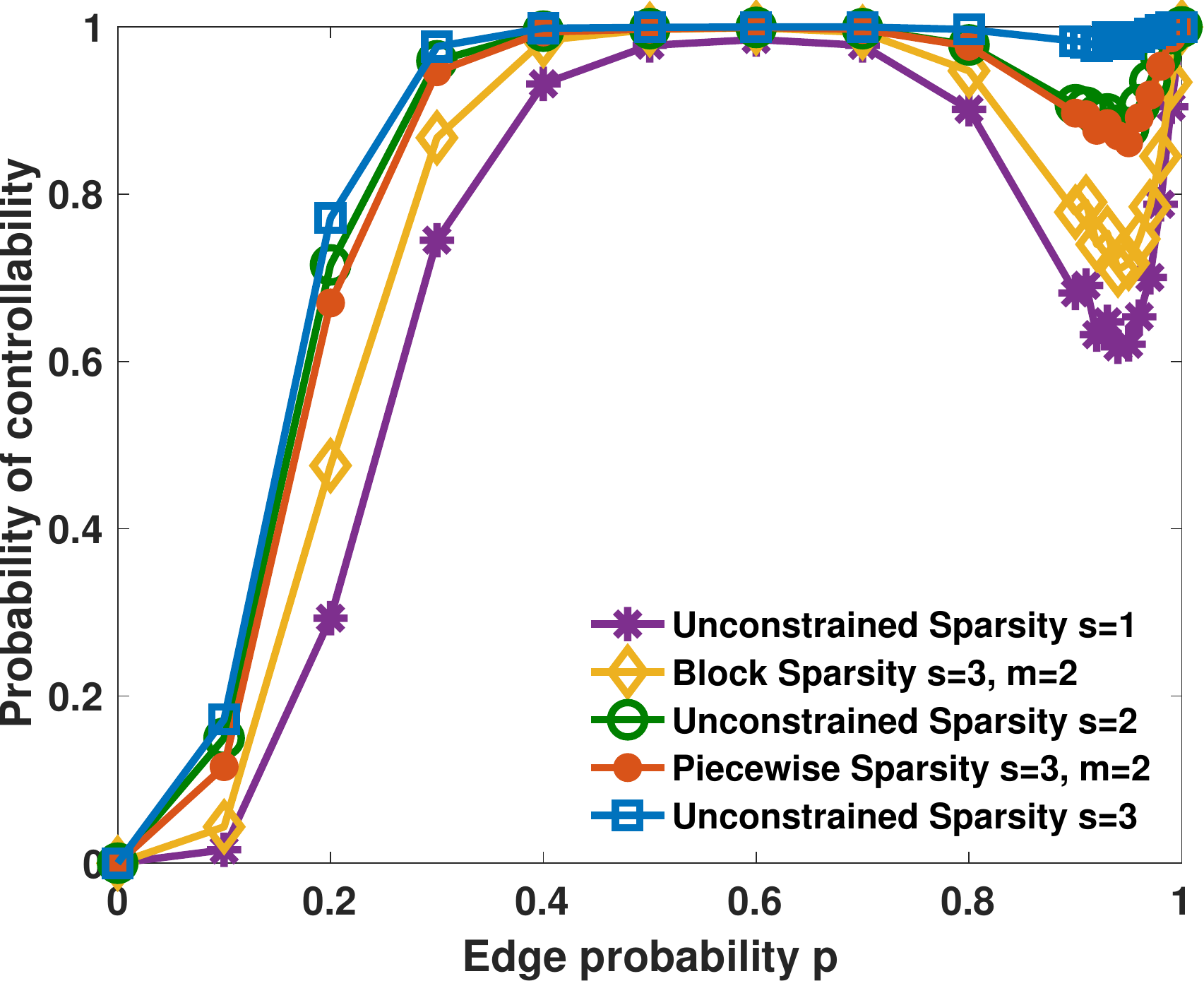}
\label{fig:p_undirected12}
}
\subfloat[Directed graph with $N=12$.]
{
\includegraphics[width=5.7cm]{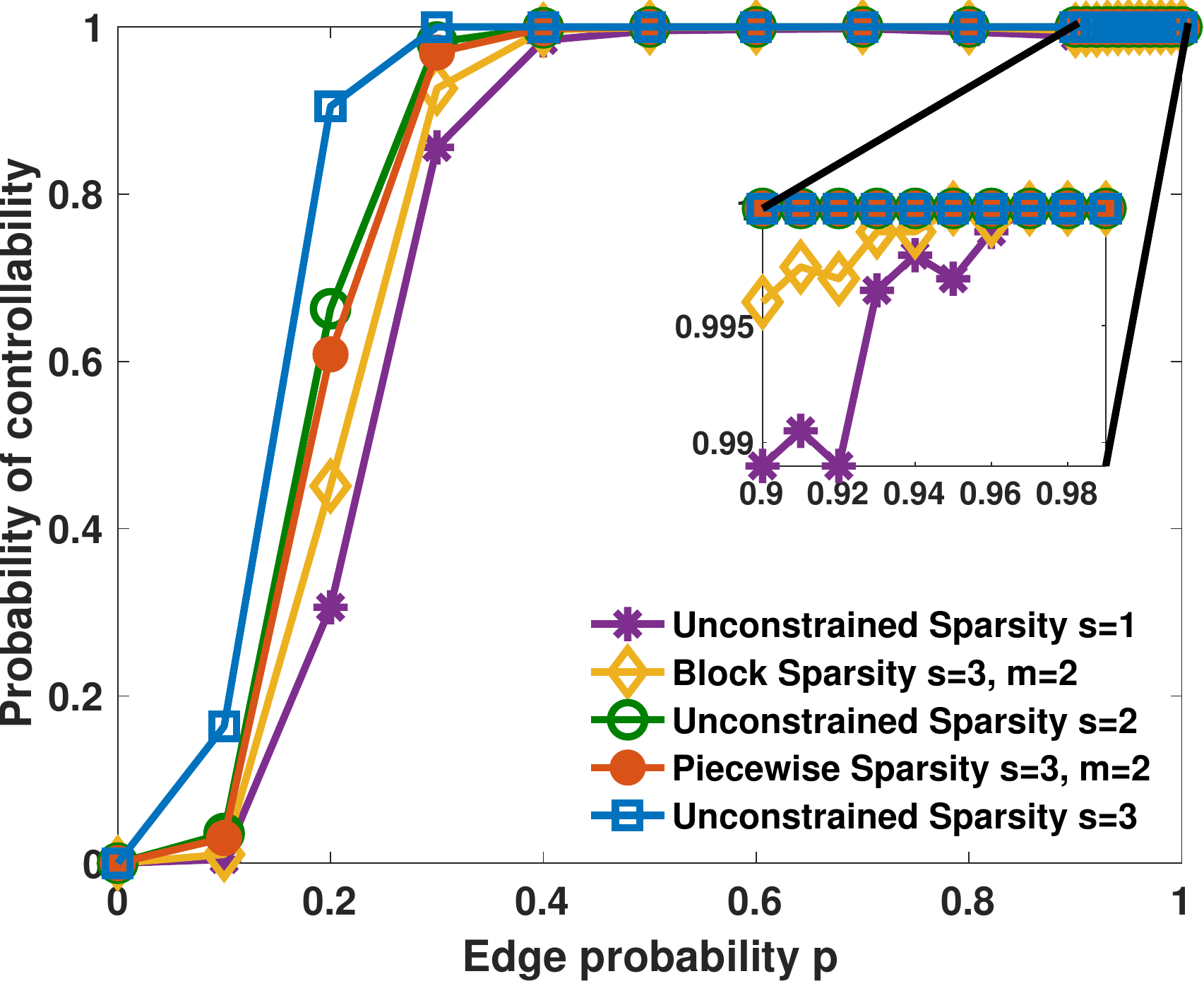}
\label{fig:p_directed12}
}

\vspace{-0.3cm}
\subfloat[Undirected graph with $N=20$.]
{
\includegraphics[width=5.7cm]{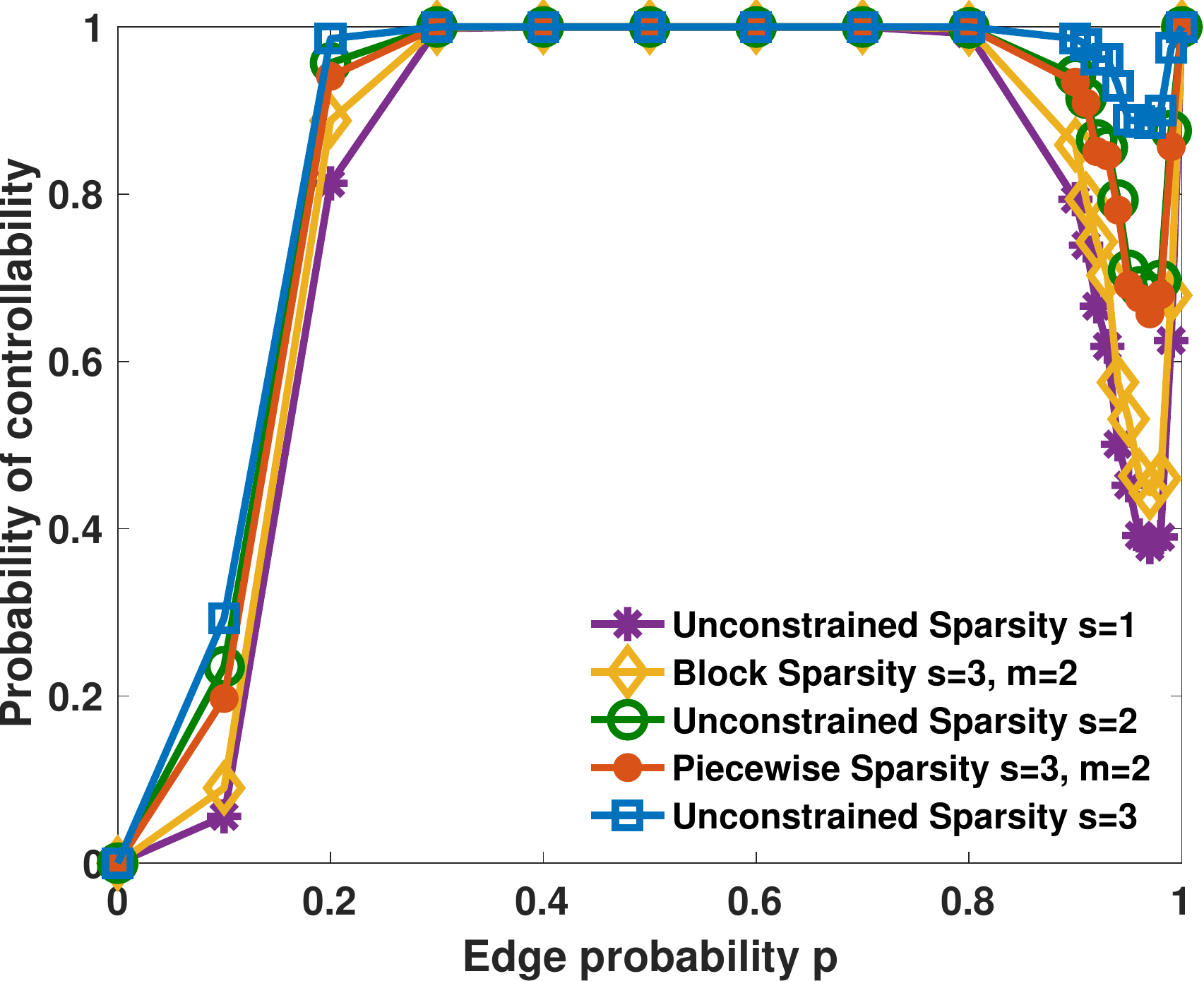}
\label{fig:p_undirected20}
}
\subfloat[Directed graph with $N=20$.]
{
\includegraphics[width=5.7cm]{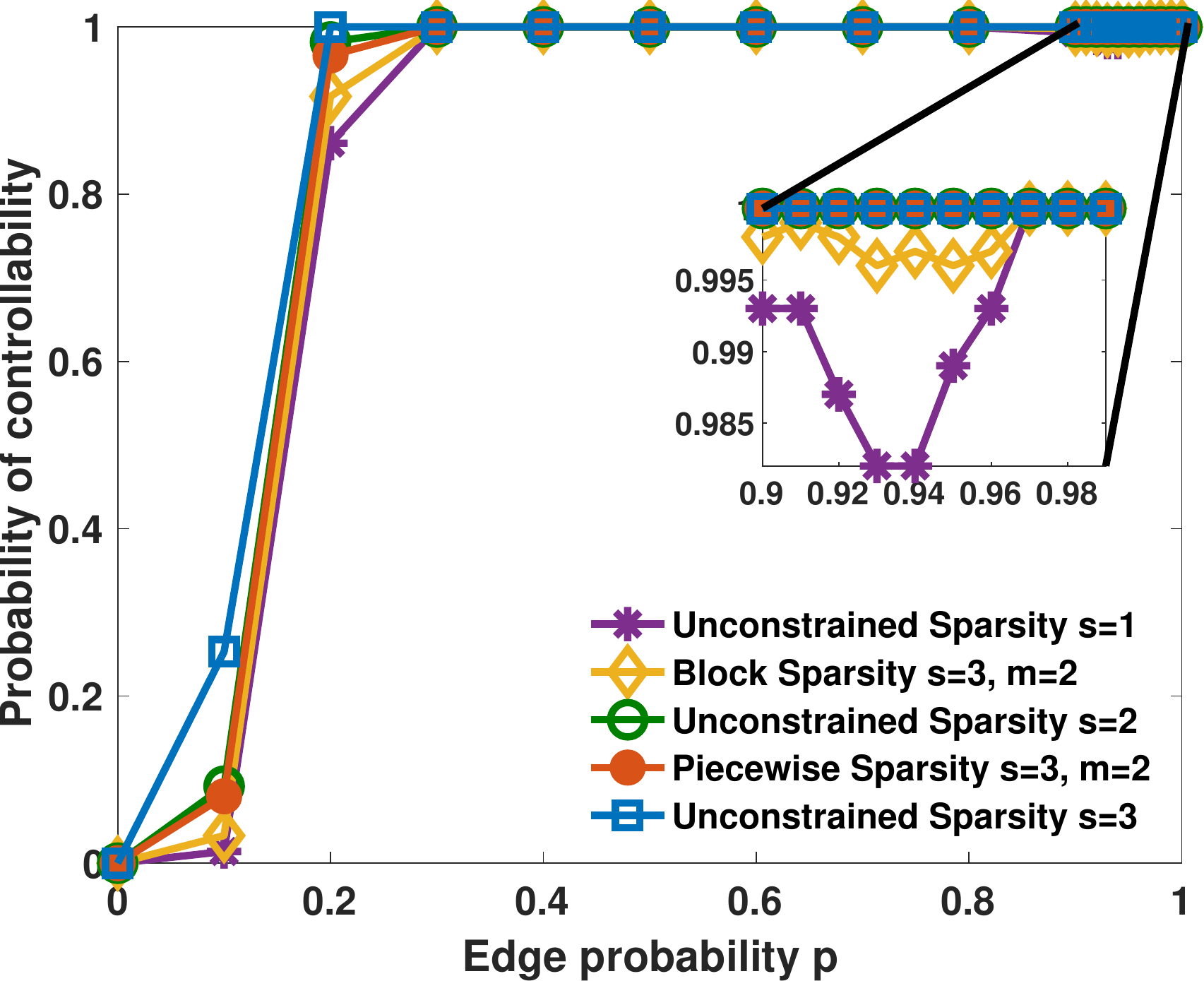}
\label{fig:p_directed20}
}

\vspace{-0.3cm}
\subfloat[Undirected graph with $N=50$.]
{
\includegraphics[width=5.7cm]{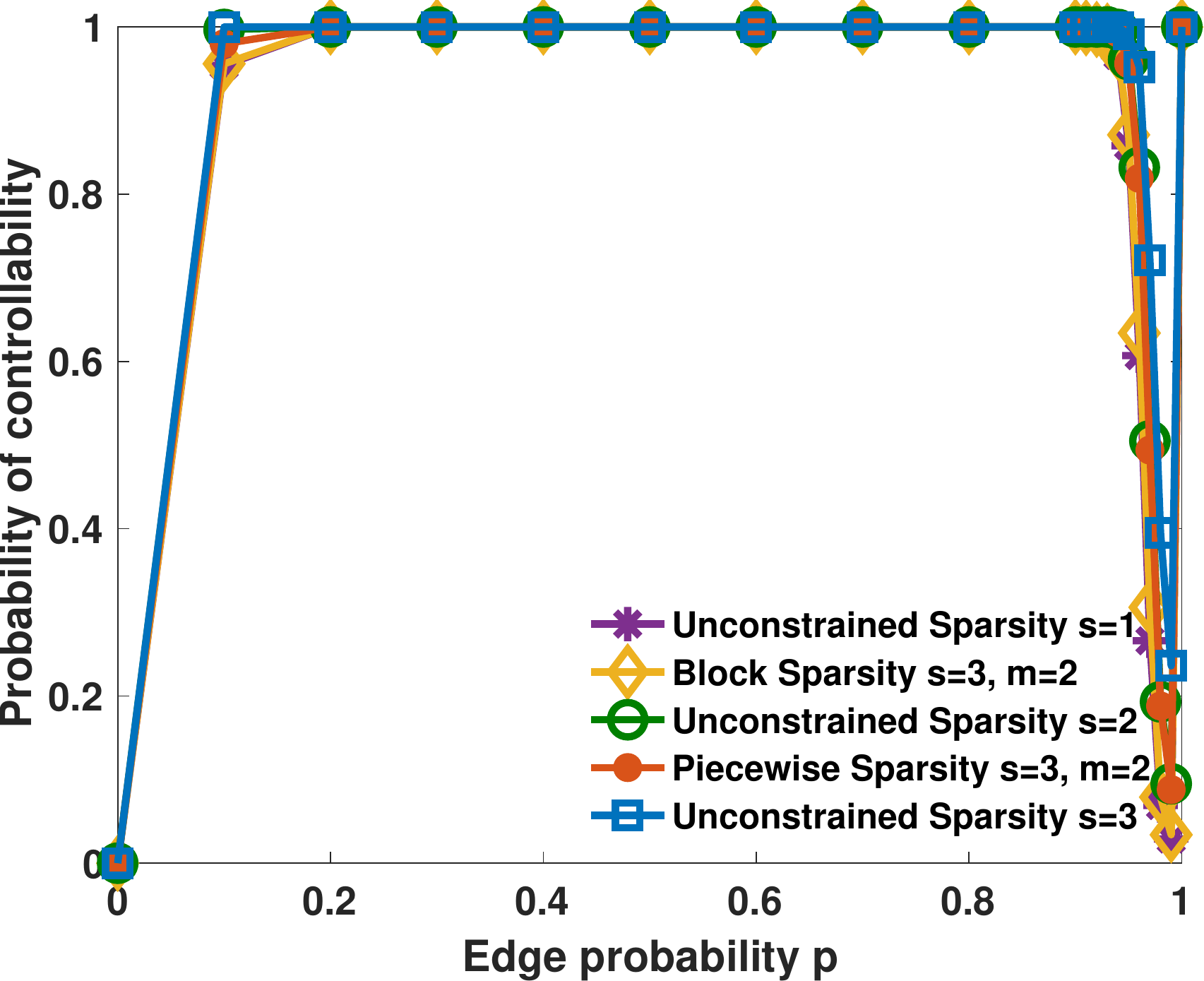}
\label{fig:p_undirected50}
}
\subfloat[Directed graph with $N=50$.]
{
\includegraphics[width=5.7cm]{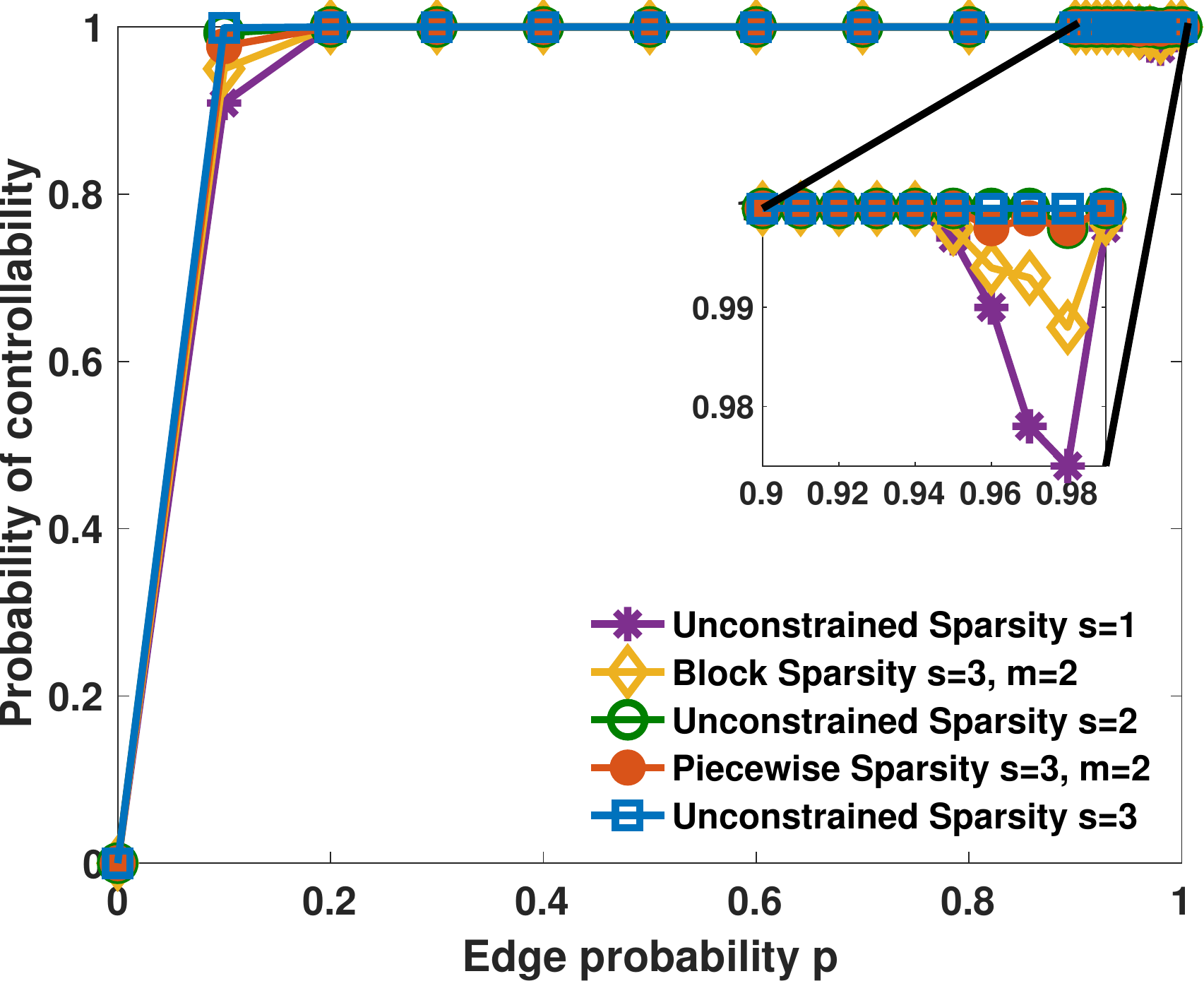}
\label{fig:p_directed50}
}
\end{center}

\vspace{-0.3cm}
\caption{Variation of the probability of controllability of the network opinion with edge probability $p$. The figures show that the probability of controllability grows with sparsity $s$ and its variation with edge probability $p$ is lower bounded by the relationship given in \Cref{thm:undirected,thm:directed}.}
\label{fig:varyp}
\vspace{-0.7cm}

\end{figure}

\begin{figure}
\begin{center}
\subfloat[Undirected graph with $p=0.2$.]
{
\includegraphics[width=5.7cm]{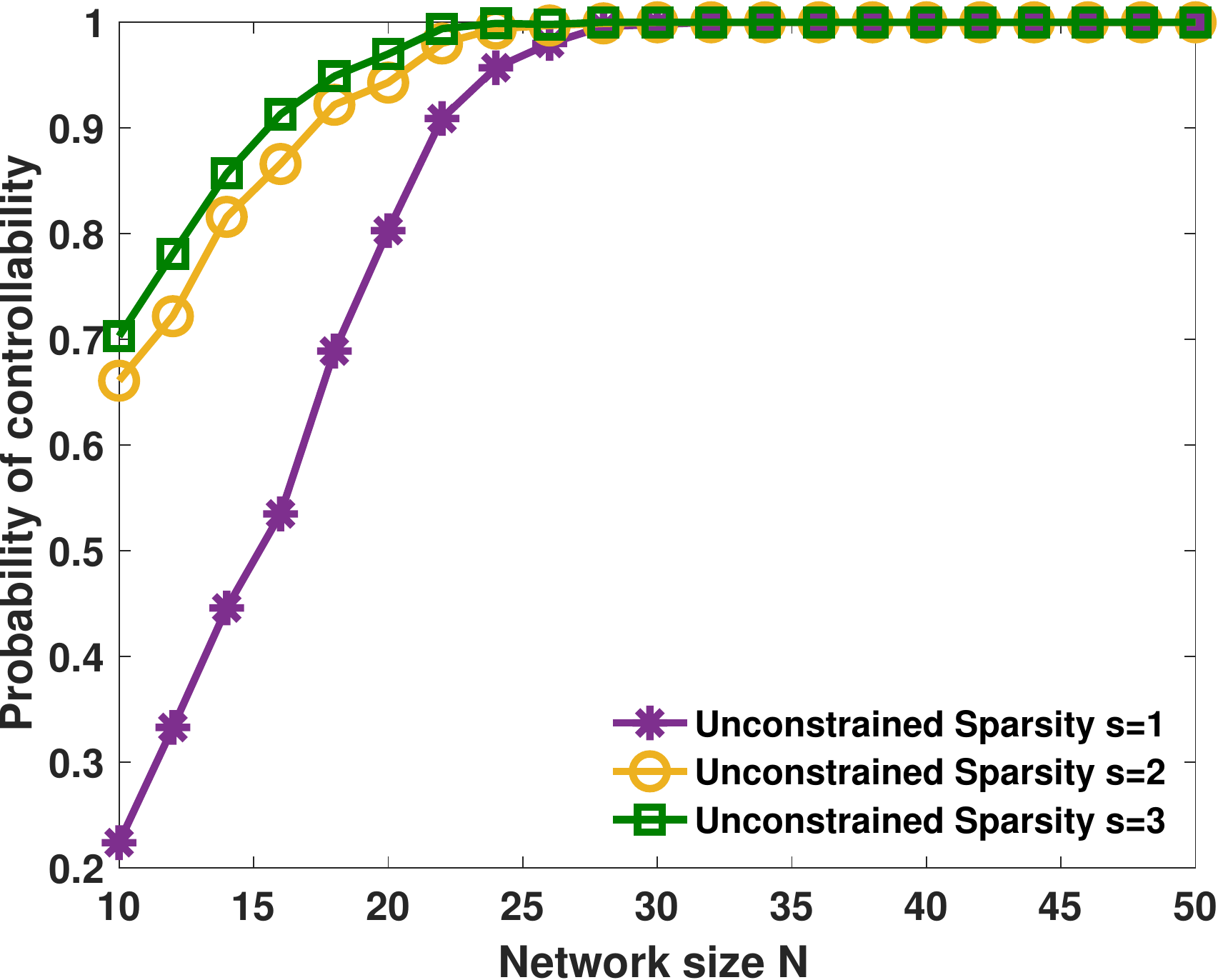}
\label{fig:N_undirected}
}
\subfloat[Directed graph with $p=0.2$.]
{
\includegraphics[width=5.7cm]{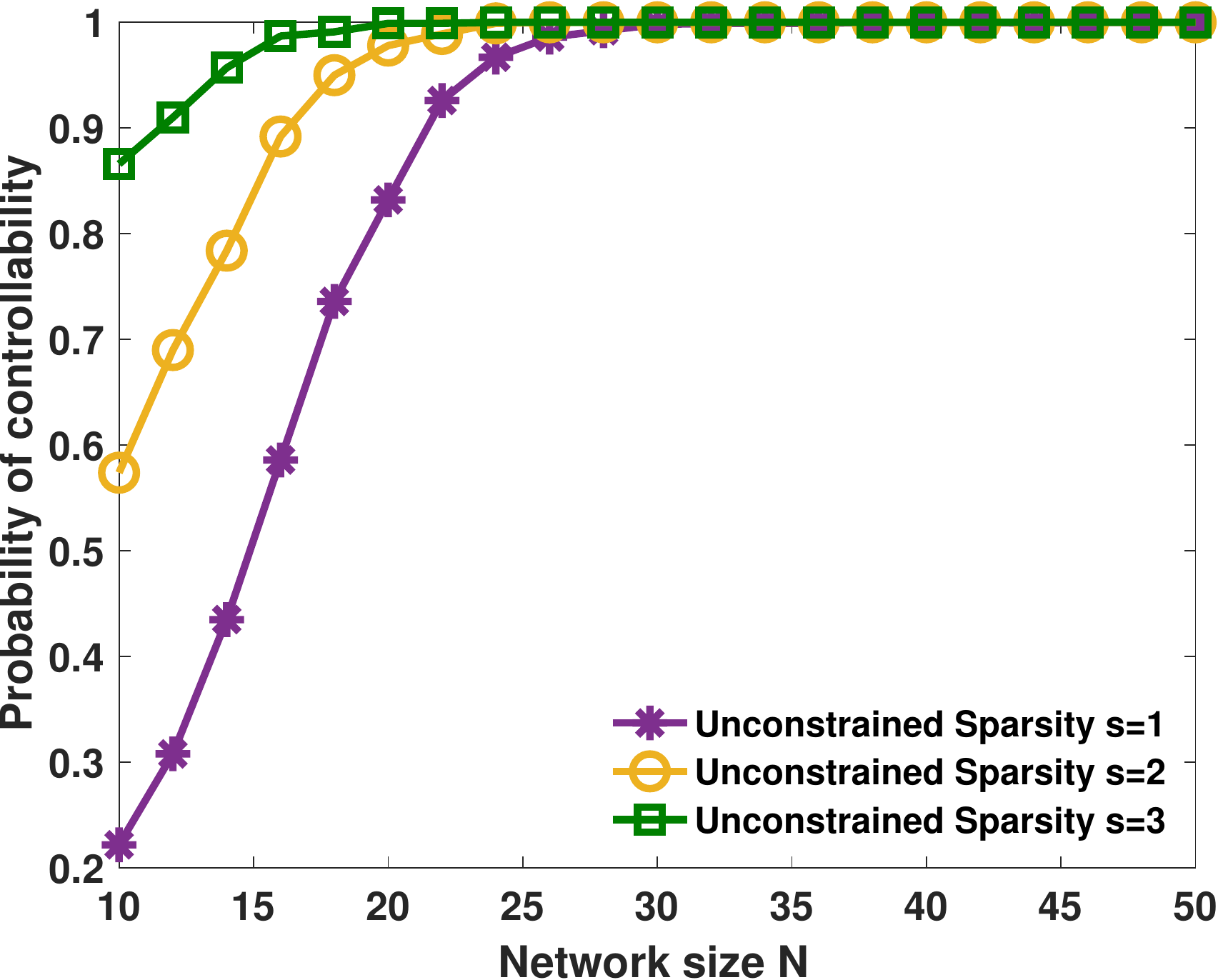}
\label{fig:N_directed}
}
\end{center}

\vspace{-0.3cm}
\caption{Variation of the probability of controllability of the network opinion with network size $N$. The figures confirm that the probability of the network opinion being not controllable decreases exponentially with the network size $N$, as given by \Cref{thm:undirected,thm:directed}.}
\label{fig:varyN}
\vspace{-0.7cm}

\end{figure}

\begin{figure}
\begin{center}
\subfloat[Undirected power-law graph with $N=24$.]
{
\includegraphics[width=5.7cm]{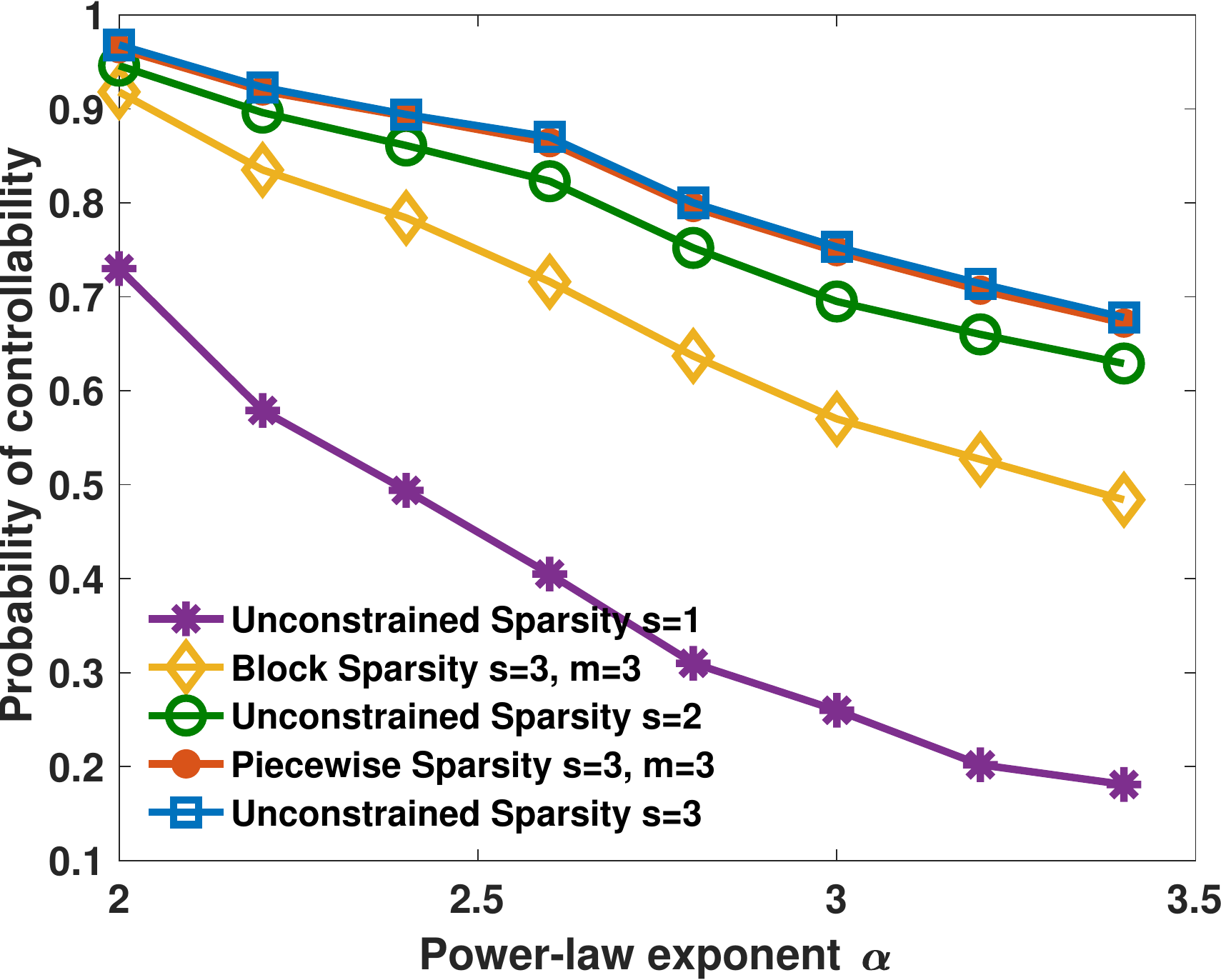}
\label{fig:alpha_undirected_powerlaw}
}
\subfloat[Undirected power-law graph with $\alpha =2.5$.]
{
\includegraphics[width=5.7cm]{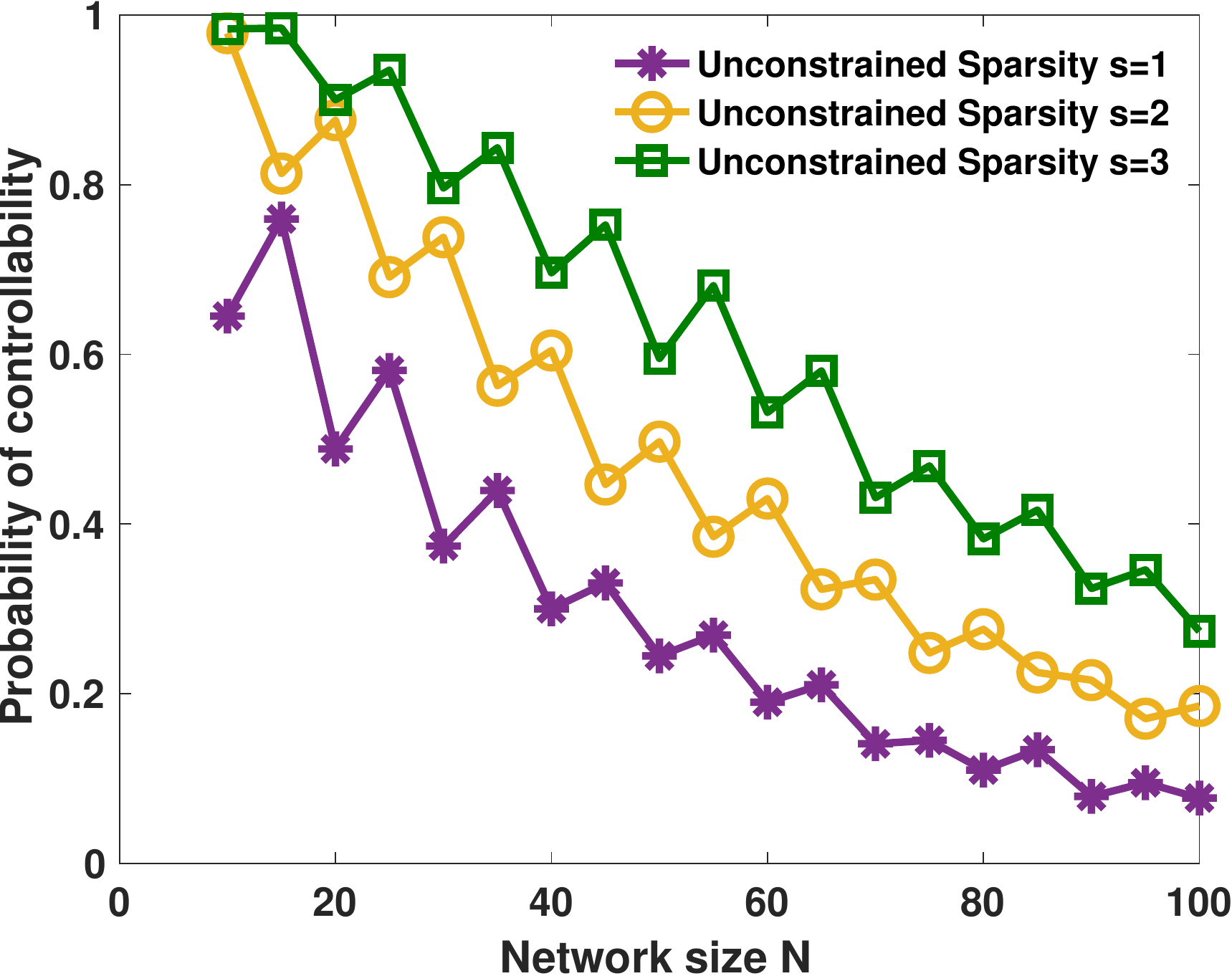}
\label{fig:N_undirected_powerlaw}
}
\end{center}
\caption{Variation of the probability of controllability of the network opinion of the power-law graph model with the power-law exponent~$\alpha$ and the network size $N$. The figures show that the probability of controllability in the power-law model is significantly different from that of the \ER model.}
\label{fig:powerlaw}
\end{figure}

To give additional insights, we compute the probability of controllability of network opinion using \Cref{thm:controlpattern} via numerical experiments, and compare the results with the bounds in \Cref{thm:undirected,thm:directed}.  Also, we numerically evaluate the probability of controllability for a different model of social networks called a power-law model~\cite{barabasi1999emergence, newman2005power, mitzenmacher2004brief} in order to understand how probability of controllability varies for different models.\footnote{\change{In this paper, a power-law graph refers to an undirected graph where the probability~$p(k)$ that a uniformly sampled node has $k$ neighbors~(i.e.,~degree distribution evaluated at $k$) is proportional to $k^{-\alpha}$ for a fixed value of the power-law exponent $\alpha > 0$} \change{(the term power-law graph is used in some literature to refer to a directed graph with both in- and out-degree distributions following a power-law though we do not deal with such directed power-law graphs in this paper)}. It has been shown that power-law degree distributions arise naturally from simple and intuitive generative processes such as preferential attachment whose power-law exponent $\alpha$ lies in the range from $2$ to $3$~\cite{barabasi1999emergence, newman2005power, mitzenmacher2004brief}. Hence, they have been widely compared with \ER graphs in the social network literature~\cite{lerman2016, nettasinghe2018your}. The key difference between the two graph models (\ER and power-law) lies in the degree distribution~(\ER graphs have Poisson degree distributions as opposed to power-law degree distributions) which is a key structural property of networks with implications in epidemic spreading, stability, friendship paradox and perception bias etc.~\cite{alipourfard2020friendship, pastor2015epidemic, cohen2000resilience}.}
\subsection{Probability of controllability of the \ER model}\label{sec:ermodel}
To evaluate the probability of controllability of the \ER model,  we simulated 1000 independent realizations of both undirected and directed \ER graphs each (for each value of $N$ and $p$). The fraction of the realizations that satisfy the two conditions of \Cref{thm:controlpattern} (with $\matPhi$ as the adjacency matrix and  $\matPsi=\eye$) is the estimate of probability of the opinion being controllable. 

As we mentioned in \Cref{sec:practical}, the \ER model captures the unknown structure of the underlying social network. Several real world networks such as high-school romantic partner networks have been shown to be similar to \ER model~\cite{bearman2004chains}. Further, although the \ER model might not capture all the characteristics of other social networks, it provides the simplest and most analytically tractable approximation for such networks (for example, the emergence of giant connected components)~\cite{jackson2010social}. In this context, our numerical results in this section help to better understand the effect of the parameters of the \ER model on another such sociologically important phenomena, namely controllability of opinions in social networks. 
The key observations from the numerical results are as follows:
\begin{itemize}[leftmargin=0.3cm]
\item \emph{Sparsity and PCBS model: }  \Cref{fig:varyp,fig:varyN} confirm that as sparsity $s$ increases, the probability of controllability grows. This trend is in agreement with the bounds in \Cref{thm:undirected,thm:directed} which also capture the monotonically increasing nature of the probability of controllability with $s$. Also, for $s=3$, the unconstrained sparse vectors offer the highest probability of controlling network opinion, followed by the piece-wise and block sparse vectors. This order verifies  the relation given by \eqref{eq:sparsity_pattern_order}.

\item \emph{Edge probability $p$: } \Cref{fig:varyp} shows that the probability of controllability first increases with $p$, reaches its maximum value, and then decreases. Also, the probability of controllability is one when $p$ is close to 1. For comparison, we note that the bounds on the probability of controllability (given by \Cref{thm:undirected,thm:directed})  approximately scale as $(1-p)^{Ns}(1-\exp((pN)^{\alpha})$, for $\alpha>0$. This bound is zero when $p=0$, then increases with $p$ to attain a maximum value and diminishes thereafter. Thus, both the bounds in \Cref{thm:undirected,thm:directed} and the curves in \Cref{fig:varyp} show similar behaviors. However, as $p$ approaches 1, the bound decreases, whereas the probability of controllability estimated in \Cref{fig:varyp} suddenly increases when $p=1$. This difference in behavior is because the values of $p$ close to 1 lie outside the regime of the edge probability for which \Cref{thm:undirected,thm:directed} hold. Also, this change in probability of controllability is not surprising because when $p=1$, the adjacency matrix becomes $\matA=\one\one^\T-\eye$ which is a deterministic full rank matrix. Therefore, both the conditions of \Cref{thm:controlpattern} are satisfied by the system for all values of $s$ and all sparsity patterns. Hence, the probability of controllability is $1$. 

\item \emph{Network size $N$: } \Cref{fig:varyN} indicates that as the network size $N$ grows, the probability of the system not being controllable decreases exponentially. This observation corroborates the dependence of $N$ on the probability of controllability given by \Cref{thm:undirected,thm:directed}.  Also, \Cref{thm:undirected,thm:directed} imply the opinion of an asymptotically large network is  controllable, almost surely, and the asymptotic behavior is attained in the regime $N>30$ when $p=0.2$. This observation is confirmed from \Cref{fig:varyp}  that reveals that as $N$ becomes larger, the network opinion is controllable with high probability for a wider range of edge probability $p$ values.

\item\emph{Undirected and directed graphs:} \Cref{fig:varyp,fig:varyN} show that the probability of controllability is larger for directed graphs compared to undirected graphs, in all settings. This is an additional insight which is not evident from \Cref{thm:undirected,thm:directed}.
\end{itemize}
\subsection{Probability of controllability of the power-law model}\label{sec:powerlaw}
The aim of this subsection is to show that power law networks behave very differently from \ER networks.  Recall that an \ER network has a Poisson degree distribution, whereas a power-law network has a degree distribution of the form~${p(k) = Ck^{-\alpha}}$ where $C$ is the normalizing constant and $\alpha>0$ is the power-law exponent. The simulation results presented below for power-law networks show that our theoretical results do not hold for this case, and there is a strong motivation to extend the results of this paper to other random graph models in future work.

To evaluate the probability of controllability for a power-law model, we simulated 1000 independent realizations of undirected power-law graphs using the so called configuration model~\cite{newman2003structure} (for each value of the network size~$N$ and power-law exponent~$\alpha$). More specifically, the configuration model generates $k$~half-edges for each of the $N$ nodes in the graph where~$k$ is the number obtained by rounding the realizations sampled independently from the power-law distribution~i.e.,~$k \sim Ck^{-\alpha}$ where $C$ is the normalizing constant and $\alpha>0$ is the power-law exponent. Then, each half-edge is connected to another randomly selected half-edge avoiding parallel edges and self-loops, yielding a graph with a power-law degree distribution. Finally, the fraction of the realizations that satisfy the conditions of \Cref{thm:controlpattern} is used as the estimate of the probability of controllability. The results are presented in \Cref{fig:powerlaw} where the definition of the labels are the same as those in \Cref{fig:varyp,fig:varyN} (see \Cref{sec:ermodel}).

 \Cref{fig:alpha_undirected_powerlaw} shows that the probability of controllability decreases monotonically with power-law exponent $\alpha$ for all considered sparsity models. This observation is intuitive because a smaller power-law exponent $\alpha$ implies that the network has larger number of high-degree nodes, making it easier to control. This is different from the non-monotone relation observed in \Cref{fig:varyp} for \ER graphs. However, it should also be noted that the parameter $p$ of the \ER model and the parameter $\alpha$ of the power-law model convey different information: $p$ is the probability of the presence of an edge whereas $\alpha$ is directly related to the degree of nodes. 
Further, \Cref{fig:N_undirected_powerlaw} shows that the probability of controllability decreases with the number of nodes $N$ in power-law model indicating an opposite behavior to the \ER model shown in \Cref{fig:N_undirected}. Also, unlike the \ER model, the variation of the probability with $N$ is not smooth. However, the cause of non-smoothness of the curve is not obvious, and we defer it as future work. To sum up, these differences suggest that the probability of controllability is an inherent property of the model. 

\section{Conclusion}
This paper analyzed controllability of network opinions modeled using a linear propagation framework with the additive influence of a sparsity constrained manipulative agent. The linear propagation was modeled using an \ER graph for two cases: the undirected and directed graphs. At every time instant, the  agent can influence only a small (compared to the network size) number of people chosen according to a predefined sparsity pattern. The main results were \Cref{thm:undirected} and \Cref{thm:directed} for undirected and directed graphs, respectively. They provide lower bounds on the probability with which the manipulative agent is able to drive the network opinion to any desired state starting from an arbitrary network opinion. Our results indicate that in both cases, the probability increases with the network size, and the opinions on an asymptotically large network is almost surely controllable. 

One limitation of our results (\Cref{thm:undirected,thm:directed}) is that they are useful only if $s\ll N$. Generalizing the results for all values of sparsity $s$ is deferred to future work. Also,  relaxing the rank one assumption on the weight matrix 
and exploring the controllability of opinions on other random graph models~(e.g.~power-law model, stochastic block model) also remain as interesting future directions. 
\newpage

\appendix
\crefalias{section}{appendix}
\crefalias{subsection}{appendix}

\section{Relation between Piece-wise and Block Sparsity Models}\label{app:relation}
In this section, we assert that the block sparse vectors can be seen as a special case of the piece-wise sparse vectors. This observation implies that if a system is controllable using block sparse control inputs, then it is controllable using piece-wise sparse control inputs (see \Cref{sec:dependence}). 

For any given block sparse vector $\vecz\in\bbR^N$,  we can rearrange its entries as follows:
\begin{center}
\begin{blockarray}{c[c]}
		\ldelim\{{4}{*}[Block $1$] & $\vecz_1$  \\
		&$\vecz_2$\\
		&$\vdots$ \\
		&$\vecz_m $\\
		 \ldelim\{{3}{*}[Block $2$]&$\vecz_{m+1}$ \\
		&$\vdots$ \\
		&$\vecz_{2m}$ \\
		$\vdots$& $\vdots$ \\
		\ldelim\{{3}{*}[Block $\frac{N}{m}$]&$\vecz_{N-m+1}$ \\
		&$\vdots$ \\
		&$\vecz_{N}$
	\end{blockarray} $\to$  \begin{blockarray}{[c]c}
		 $\vecz_1$  &\rdelim\}{4}{*}[Block 1] \\
		 $\vecz_{m+1}$\\
		$\vdots$ \\
		$\vecz_{N-m+1}$ \\
		 $\vecz_{2}$ &\rdelim\}{3}{*}[Block $2$] \\
		$\vdots$ \\
		$\vecz_{N-m+2}$ \\
		$\vdots$& $\vdots$ \\
		$\vecz_{m} $&\rdelim\}{3}{*}[Block $m$.] \\
		$\vdots $\\
		$\vecz_{N}$
	\end{blockarray}
	\end{center}
The rearranged vector on the right-hand side is a piece-wise sparse vector with $m$ blocks, each with at most sparsity $s/m$ and the same support. Since shuffling the entries of the control input does not change the controllability-related properties of the system in \eqref{eq:system_model}, controllability under block sparsity can be seen as a special case of that under piece-wise sparsity. We illustrate this idea using the example below:

\begin{example} 
Consider the case where $N=6, s=2$ and $m=2$. A block sparse vector with these parameters has $N/m=3$ blocks out of which only $s/m=1$ block has nonzero entries. \change{Therefore, for different choices of support, the block sparse vectors can be rearranged to piece-wise sparse vectors with  $m=2$ blocks of size $N/m=3$ where each block has at most $s/m=1$ nonzero entries as shown below:
\begin{equation*}
 \begin{bmatrix}
a\in\bbR\\
b\in\bbR\\
0\\
0\\
0\\
0
\end{bmatrix}\to \begin{bmatrix}
a\in\bbR\\
0\\
0\\
b\in\bbR\\
0\\
0
\end{bmatrix},\hspace{1cm}\begin{bmatrix}
0\\
0\\
a\in\bbR\\
b\in\bbR\\
0\\
0
\end{bmatrix}\to \begin{bmatrix}
0\\
a\in\bbR\\
0\\
0\\
b\in\bbR\\
0
\end{bmatrix}, \hspace{1cm}\begin{bmatrix}
0\\
0\\
0\\
0\\
a\in\bbR\\
b\in\bbR
\end{bmatrix}\to  \begin{bmatrix}
0\\
0\\
a\in\bbR\\
0\\
0\\
b\in\bbR
\end{bmatrix}.
\end{equation*}}
However, the rearranged vectors only form a small subset of the set of piece-wise sparse vectors. The excluded piece-wise sparse vectors take the following forms: 
\begin{equation*}
\lc \begin{bmatrix}
a\in\bbR\\
0\\
0\\
0\\
b\in\bbR\\
0
\end{bmatrix}, \begin{bmatrix}
a\in\bbR\\
0\\
0\\
0\\
0\\
b\in\bbR
\end{bmatrix}, \begin{bmatrix}
0\\
a\in\bbR\\
0\\
b\in\bbR\\
0\\
0
\end{bmatrix}, \begin{bmatrix}
0\\
a\in\bbR\\
0\\
0\\
0\\
b\in\bbR
\end{bmatrix}, \begin{bmatrix}
0\\
0\\
a\in\bbR\\
b\in\bbR\\
0\\
0
\end{bmatrix}, \begin{bmatrix}
0\\
0\\
a\in\bbR\\
0\\
b\in\bbR\\
0
\end{bmatrix}\rc.
\end{equation*}
Thus, the larger set of piece-wise sparse vectors is less restricted than the set of block sparse vectors.
\end{example}

\section{Proof of \Cref{thm:controlpattern}}
\label{app:control_pattern}

At a high level, the proof has three  main steps:
\begin{enumerate}[leftmargin=0.5cm]
\item \label{step:A}We first prove that controllability as defined in \Cref{thm:controlpattern} is equivalent to the following: there exist an integer $K>0$ and a matrix $\tilde{\matPsi}\in\calB_{(K)}$ such that 
\begin{equation}\label{eq:eq_control}
\rank{\tilde{\matPsi}}=N,
\end{equation} where $\calB_{(K)} = \lc \tilde{\matPsi}\in \bbR^{N\times Ks} : \tilde{\matPsi} = \begin{bmatrix}
\matPhi^{K-1}\matPsi_{\calS_1} & \matPhi^{K-2}\matPsi_{\calS_2}  \ldots \matPsi_{\calS_K}
\end{bmatrix}, \calS_k\in\calU\rc.$
\item Next, we show that when one of the conditions, either \Cref{con:control_pattern_a} or \Cref{con:control_pattern_b} of \Cref{thm:controlpattern} does not hold, the condition given in Step \ref{step:A} is violated. This is equivalent to showing that \Cref{con:control_pattern_a,con:control_pattern_b} of \Cref{thm:controlpattern} are necessary for our notion of controllability to hold.
\item Finally, we show that when the condition given in Step \ref{step:A} does not hold, \Cref{con:control_pattern_a,con:control_pattern_b} of \Cref{thm:controlpattern} are not true simultaneously. Thus, we show the sufficiency part of the result.
\end{enumerate}
We present the proof for the above steps in the following subsections:
\subsection{An equivalent condition}
To characterize controllability of the system as defined in \Cref{thm:controlpattern}, we consider the following equivalent representation:
\begin{align*}
\vecalpha_{K}-\matPhi^K\vecalpha_0 = \sum_{k=1}^K \matPhi^{K-k}\matPsi\vecv_k=\sum_{k=1}^K \matPhi^{K-k}\matPsi_{\calS_k}\vecv_{k,\calS_k},
\end{align*}
where $\calS_k\in\calU$ is the support of $\vecv_k$ and $\matPsi_{\calS_k}\in\bbR^{N\times\lv\calS_k\rv}$ is the submatrix of $\matPsi$ with columns indexed by $\calS_k$. Therefore, the system is controllable as defined in \Cref{thm:controlpattern} iff the set $\calW_{(K)}=\bbR^N$ for some finite $K$ with
$\calW_{(K)} \triangleq \cup_{\tilde{\matPsi}\in\calB_{(K)}}\;\range{\tilde{\matPsi}}$,
where $\range{\cdot}$ denotes the column space of a matrix. However, a vector space over an infinite field ($\bbR^N$ in this case) cannot be a finite union of its proper subspaces~\cite[Chapter 1]{friedland2018linear}. Therefore, $\range{\tilde{\matPsi}}=\bbR^N$, for some $\tilde{\matPsi}\in\calB$, and {Step \ref{step:A} in the proof outline is completed}.
\subsection{Necessity}
We consider the following two cases:
\begin{enumerate}[label=(\roman*),leftmargin=0.5cm]
\item Suppose that \Cref{con:control_pattern_a} in  \Cref{thm:controlpattern} does not hold. Then, from the classical PBH test  for controllability, the linear dynamical system defined by the state transition matrix-input matrix pair $(\matPhi,\matPsi_{\calM})$ is not controllable. Therefore, controllability matrix $\tilde{\matPsi}_{(K)}= \begin{bmatrix} \matPhi^{K-1}\matPsi_{\calM} & \matPhi^{K-2}\matPsi_{\calM} & \ldots & \matPsi_{\calM}
\end{bmatrix}$ does not have full row rank for any finite $K$.
Further, all matrices in $\calB_{(K)}$ are submatrices of $\tilde{\matPsi}_{(K)}$, and therefore, \eqref{eq:eq_control} is violated any $\tilde{\matPsi}\in\calB_{(K)}$.

\item Suppose \Cref{con:control_pattern_b} in  \Cref{thm:controlpattern} does not hold. Then, for every index set $\calS\in\calU$, there exists a nonzero vector $\vecz$ such that $\vecz^\T\matPsi_{\calS}=\zero$ and $\vecz^\T\matPhi=\zero$. This implies that for any finite $K$, there exists a vector $\vecz$ such that $\vecz^\T\tilde{\matPsi}=\zero$, for all $\tilde{\matPsi}\in\calB_{(K)}$. Therefore, \eqref{eq:eq_control} is violated any $\tilde{\matPsi}\in\calB_{(K)}$.
\end{enumerate}
Hence, we proved the necessity of the conditions given by \Cref{thm:controlpattern}.
\subsection{Sufficiency}
Suppose that \eqref{eq:eq_control} is not true for any integer $K>0$ and $\tilde{\matPsi}\in\calB_{(K)}$. We consider
\begin{multline*}
\tilde{\matPsi}^* =
[ \begin{matrix}
\matPhi^{ P  N-1 }\matPsi_{\calS_1} &  \matPhi^{  P  N-2} \matPsi_{\calS_1} &\ldots 
  & \matPhi^{( P -1 ) N} \matPsi_{\calS_1}  \end{matrix}\\
\begin{matrix}
\ldots & \matPhi^{( P -1) N-1}\matPsi_{\calS_2}&\ldots \matPhi^{( P -2) N}\matPsi_{\calS_2} &\ldots
\end{matrix}\\
\begin{matrix}
\ldots &  \matPhi^{N-1}\matPsi_{\calS_{P}}&\ldots&\matPsi_{\calS_{P}}
\end{matrix}],
\end{multline*}
where $P=\lv\calU\rv$ and $\cup_{i=1}^P\calS_{i}=\calU$. 
Since \eqref{eq:eq_control} does not hold for any finite $K$ and  $\tilde{\matPsi}\in\calB_{(K)}$, the matrix $\tilde{\matPsi}^*$ does not have full row rank. Next, we can rearrange the columns of $\tilde{\matPsi^*}$ to get the following matrix which has the same rank as that of $\tilde{\matPsi}^*$:
$\begin{bmatrix}
\matPhi^{N-1}\matPsi^* & \matPhi^{N-2}\matPsi^* & \ldots & \matPsi^*
\end{bmatrix}$, where $\matPsi^*\in\bbR^{N\times  P s}$ is defined as 
$\matPsi^*\triangleq\begin{bmatrix}
\matPhi^{( P -1)N}\matPsi_{\calS_1} & 
\matPhi^{( P -2)N}\matPsi_{\calS_2}
\ldots &
\matPsi_{\calS_{ P }}
\end{bmatrix}$.
Thus, using classical Kalman rank test for controllability without any constraints, the linear dynamical system defined by the state transition matrix-input matrix pair $(\matPhi,\matPsi^*)$ is not controllable. Then, the classical PBH test for controllability without any constraints, implies that the matrix $\begin{bmatrix}
\matPhi-\lambda\eye & \matPsi^*
\end{bmatrix}\in\bbR^{N\times N+\tilde{K}s}$ has rank less than $N$, for some $\lambda\in\bbC$. Therefore, there exists a vector $\vecz\neq \zero\in\bbR^N$ such that $\vecz^\T\matPhi=\lambda\vecz^\T$ and $\vecz^\T\matPsi^*=\zero$. However, we have
\begin{equation*}
\zero  = \vecz^\T\matPsi^*  = \vecz^\T \begin{bmatrix}
\lambda^{ (P-1)N} \matPsi_{\calS_1} & 
\lambda^{ (P-2)N} \matPsi_{\calS_2}
\ldots &
\matPsi_{\calS_{P}}
 \end{bmatrix}.
\end{equation*}
So either $\lambda=0$ and $\vecz^\T\matPsi_{\calS_{P}}=\zero$, or, if $\lambda$ is nonzero, then  $\vecz^\T\matPsi_{\calM}=\zero$. Since the ordering of the index sets in $\calU$ does not matter, we conclude that either $\lambda=0$ and $\vecz^\T\matPsi_{\calS}=\zero$ for all $\calS\in\calS$, or,  $\vecz^\T\matPsi_{\calM}=\zero$ for some $\lambda\in\bbC$. Therefore, the two conditions of \Cref{thm:controlpattern} do not hold simultaneously. Thus, the proof is complete. 

\hfill\qed
\section{Proof of \Cref{thm:undirected}}\label{app:undirected}
\Cref{thm:controlpattern} provides necessary and sufficient conditions under which a system is controllable using sparse inputs. Therefore, the key idea of the proof of \Cref{thm:undirected} is to derive the probability with which the conditions of \Cref{thm:controlpattern} hold when $\matPhi$ and $\matPsi$ in \Cref{thm:controlpattern} are set to be $\bar{\matA}$ and $\eye$, respectively. The main tools used in the proof are the rank properties of the Hadamard product and a random symmetric binary matrix as stated below:

\begin{lemma}[Invertibility of Hadamard product]\label{lem:Hadamard}
For any matrix $\matA\in\bbR^{N\times N}$ and a vector $\vecw\in\bbR^N$, the Hadamard product $\matA\odot\lb\vecw\vecw^\T\rb$ is invertible if and only if $\matA$ is invertible and all entries of $\vecw$ are nonzero.
\end{lemma}
\begin{proof}
\change{Let $\bar{\matW}\in\bbR^{N\times N}$ be a diagonal matrix with the entries of $\vecw$ along its diagonal. Then it follows that
$\matA\odot(\vecw\vecw^\T) = \bar{\matW}\matA\bar{\matW}$,
where $\odot$ denotes the Hadamard product. Therefore, 
\begin{equation*}
\det\lc\matA\odot(\vecw\vecw^\T)\rc = \det\lc\bar{\matW}\rc^2\det\lc\matA\rc =\det\lc\matA\rc\prod_{i\in[N]}\vecw_i^2.
\end{equation*}
Thus, $\det\lc\matA\odot(\vecw\vecw^\T)\rc\neq 0$ if and only if $\det\lc\matA\rc\neq0$ and all the entries of $\vecw$ are nonzero. Hence, the proof is complete.}
\end{proof}

\begin{theorem}\label{thm:rankundirected}
Let $\matA\in\lc0,1\rc^{N\times N}$ be the adjacency matrix of an undirected \ER graph with the edge probability $p$. Then, there exist finite positive constants $C$ and $c$ such that, for
$N^{-1}\leq p\leq 1-N^{-1}$,
the following holds:
\begin{equation*}
\bbP\lc\matA \text{ is non-singular}\rc\geq 1-C\exp\lb-c(pN)^{1/32}\rb.
\end{equation*}
\end{theorem}
\begin{proof}
See \Cref{app:rank_undirected}.
\end{proof}

Clearly, \Cref{con:control_pattern_a} of \Cref{thm:controlpattern} holds with probability 1. So we focus on \Cref{con:control_pattern_b} of \Cref{thm:controlpattern}. Let event $\calE$ denote that event that \Cref{con:control_pattern_b} of \Cref{thm:controlpattern} holds, i.e., $\calE$ is given by
\begin{equation}\label{eq:E_defn}
\calE = \lc\exists \calS\in\calU:\rank{\begin{bmatrix}
\bar{\matA} & \eye_{\calS}
\end{bmatrix}}=N\rc.
\end{equation} 
In the following, we derive a lower bound on the probability $\bbP\lc\calE\rc$ which is also a lower bound on the probability with which the network opinion is controllable under given constraints.

For a given index set $\calS\in\calU$, we rearrange the columns of the matrix in \eqref{eq:E_defn} as
\begin{equation*}
\rank{\begin{bmatrix}
\bar{\matA} & \eye_{\calS}
\end{bmatrix}} = \rank{\begin{bmatrix}
\bar{\matA}_{\calS,:} & \eye\\
\bar{\matA}_{\calS^c,:} & \zero
\end{bmatrix}},
\end{equation*}
where $\calS^c=[N]\setminus\calS$ and $\lv\calS^c\rv=N-s$. Also, $\bar{\matA}_{\calS,:}\in\bbR^{ s \times N}$ and $\bar{\matA}_{\calS^c,:}\in\bbR^{ N-s\times N}$ are the submatrices of $\bar{\matA}$ formed by rows indexed by $\calS$ and $\calS^c$, respectively. Consequently, \eqref{eq:E_defn} can be further simplified as follows:
$\calE = \lc\exists \calS\in\calU:
\bar{\matA}_{\calS^c,:}
 \text{ is full row rank}\rc$.

We note that $\calE$ depends on the rank of a non-square matrix $\bar{\matA}_{\calS^c}$. However, since \Cref{lem:Hadamard} and \Cref{thm:rankundirected} deals with the invertibility of square matrices, we first lower bound $\bbP\{\calE\}$ in terms of probabilities with which certain square matrices are invertible. For this, we notice that
\begin{equation}
\calE \supseteq\lc\exists \calS\in\calU,\calI\subseteq\calS:\bar{\matA}_{\calI^c,:}\text{ is full row rank}\rc\label{eq:E_relax1}.
\end{equation}
where $\calI^c=[N]\setminus \calI\supseteq\calS^c$, and $\bar{\matA}_{\calI^c,:}\in\bbR^{N-\lv\calI\rv\times N}$ is the submatrix of $\bar{\matA}$ formed by rows indexed by $\calI^c$. Here, \eqref{eq:E_relax1} follows because if  all rows of $\bar{\matA}_{\calI^c,:}$ are linearly independent, then, all rows of the submatrix $\bar{\matA}_{\calS^c,:}$  of $\bar{\matA}_{\calI^c,:}$ are also linearly independent. Next, we further bound \eqref{eq:E_relax1} as follows:
\begin{equation}\label{eq:E_relax2}
\calE\supseteq \Big\{\exists \calS\in\calU,\calI\subseteq\calS:\bar{\matA}_{\calI,:}=\zero\text{ and }\bar{\matA}_{\calI^c,\calI^c} \text{ is non-singular}\Big\},
\end{equation}
where $\bar{\matA}_{\calI,:}\in\bbR^{\lv\calI\rv\times N}$ is the submatrix of $\bar{\matA}$ formed by rows indexed by $\calI$, and  $\bar{\matA}_{\calI^c,\calI^c}\in\bbR^{N-\lv\calI\rv\times N-\lv\calI\rv}$ is the (symmetric) principal submatrix of $\bar{\matA}$ formed by the rows indexed by $\calI^c$ and the corresponding columns. Therefore, we have
\begin{equation*}
\bbP\lc\calE\rc\geq\bbP \Big\{\exists \calS\in\calU,\calI\subseteq\calS:\bar{\matA}_{\calI,:}=\zero\text{ and }\bar{\matA}_{\calI^c,\calI^c} \text{ is non-singular}\Big\}.
\end{equation*}
Hence, $\bbP\{\calE\}$ now depends on the invertibility of the symmetric square matrix $\bar{\matA}_{\calI^c,\calI^c}$. 

Next, \change{we note that $
\bar{\matA}=\matLambda(\matA\odot\matW)$ where the invertible diagonal matrix $\matLambda\in\bbR_+^{N\times N}$ normalizes the rows of $\matA\odot\matW$ and $\matLambda_{ii}=\min\lc 1,
\frac{1}{\sum_{j=1}^N\matA_{ij}}\rc$. 
Since} $\matLambda$ is an invertible diagonal matrix, we deduce that
\begin{equation*}
\bbP\lc\calE\rc\geq\bbP \Big\{\exists \calS\in\calU,\calI\subseteq\calS:\matA_{\calI,:}\odot\matW_{\calI,:}=\zero\text{ and }\matA_{\calI^c,\calI^c}\odot\matW_{\calI^c,\calI^c}\text{ is non-singular}\Big\}.
\end{equation*}
The entries of $\matW_{\calI^c,\calI^c}$ are sampled from a continuous distribution, they are nonzero with probability one. Thus, \Cref{lem:Hadamard} leads to the following:
\begin{equation}
\bbP\{\calE\} \geq\bbP\Big\{\exists \calS\in\calU,\calI\subseteq\calS:\matA_{\calI,:}=\zero
\text{ and }\matA_{\calI^c,\calI^c}\text{ is non-singular} \Big\}\label{eq:prob_E}.
\end{equation}

Further, \eqref{eq:prob_E} can be further simplified using \Cref{thm:rankundirected}. For this, we rewrite the right-hand side of \eqref{eq:prob_E}  as $
\bbP\{\calE\} \geq\bbP\lc\cup_{i=0}^{s}\calE_i\rc,$
where we define $\calE_i$ as follows:
\begin{equation}\label{eq:Ei_defn}
\calE_i\triangleq \Big\{\exists \calS\in\calU,\calI\subseteq\calS:\lv\calI\rv=i,\matA_{\calI,:}=\zero
\text{ and }\matA_{\calI^c,\calI^c}\text{ is non-singular}\Big\}.
\end{equation}
However, when $\matA_{\calI^c,\calI^c}$ is invertible, all rows of $\matA$ indexed by $\calI^c$ are nonzero. Therefore, $\calE_i$ denote the event that $\matA$ has exactly $i$ zero rows (indexed by $\calI$), and the remaining rows are linearly independent. Consequently, the events $\lc\calE_i\rc_{i=0}^s$ are disjoint, and so the union bound holds with equality. Therefore, we obtain
\begin{equation}\label{eq:E_relax3}
\bbP\{\calE\} \geq  \bbP\lc\bigcup_{i=0}^{s}\calE_i\rc=\sum_{i=0}^{s}\bbP\lc\calE_i\rc.
\end{equation}

Now, we simplify $\bbP\lc\calE_i\rc$ by summing over all possible values of $\calI$ (corresponding to zero rows of $\bar{\matA}$) as follows:
\begin{equation}\label{eq:Ei_relax}
\bbP\lc\calE_i\rc= \sum_{\substack{\calI\subseteq\calS:\calS\in\calU,\\\lv\calI\rv=i}}\!\bbP\lc\matA_{\calI,:}=\zero\rc\bbP\lc\matA_{\calI^c,\calI^c}\text{ is non-singular}\rc,
\end{equation}
which we obtain  using the fact that the entries of $\matA_{\calI,:}$ and $\matA_{\calI^c,\calI^c}$ are independent.

The condition $\matA_{\calI,:}=\zero$ holds when all the independent Bernoulli variables in $\matA_{\calI,:}$ are zeros. The number of independent random variables is 
$(N-1)+(N-2)+\ldots+(N-i)=i(N-(i+1)/2)$.
Therefore, we have
\begin{equation}\label{eq:AI=0}
\bbP\lc\matA_{\calI,:}=\zero\rc=(1-p)^{i(N-(i+1)/2)}.
\end{equation}

Further, the entries of $\matA_{\calI^c,\calI^c}\in\bbR^{N-i\times N-i}$ have the same distribution as that of $\matA$. Thus, we apply  \Cref{thm:rankundirected} to get
\begin{equation}\label{eq:AIC}
\bbP\lc\matA_{\calI^c,\calI^c}\text{ is non-singular}\rc\geq 1-C\exp\lb-cp(N-i)^{1/32}\rb,
\end{equation} 
where $c>0$ is universal constant.
Combining \eqref{eq:Ei_relax}, \eqref{eq:AI=0}, and \eqref{eq:AIC}, we get that
\begin{equation}\label{eq:Ei_final}
\bbP\lc\calE_i\rc\geq Q(i,\calU)p^{i(N-(i+1)/2)} \ls1-C\exp\lb-cp(N-i)^{1/32}\rb\rs,
\end{equation}
where $Q$ is as defined in the statement of the theorem (see \eqref{eq:Q_defn}). Finally, we complete the proof by combining \eqref{eq:E_relax3} and \eqref{eq:Ei_final}.

\hfill\qed

\section{Proof of \Cref{thm:directed}}\label{app:directed}
The proof technique used here is similar to that of \Cref{thm:undirected}. However, since \Cref{thm:rankundirected} does not hold in this case, we use an an analogous theorem for directed graphs which is as follows:
\begin{theorem}\label{thm:rank_directed}
Let $\matA\in\lc0,1\rc^{N\times N}$ be the adjacency matrix of a directed \ER graph with the edge probability $p$. Let $\matD$ be a real valued diagonal
matrix independent of $\matA$ with $\lV\matD\rV \leq R\sqrt{pN}$ where $R\geq1$. Then, there exist finite positive constants $C$ and $c$ that depend on $R$ such that for
$C \frac{\log N}{N}\leq p\leq 1-C \frac{\log N}{N}$,
it holds that
$\bbP\lc\matA+\matD \text{ is non-singular}\rc\geq 1-\exp\lb-cpN\rb$.
\end{theorem}
\begin{proof}
The result is an immediate corollary of \cite[Theorem 1.11]{basak2017invertibility}.
\end{proof}

Using the arguments similar to those in the proof of \Cref{thm:undirected}, we see that all steps of the proof in  \Cref{app:undirected} until \eqref{eq:Ei_relax}  hold in this case. Hence, continuing from there, the condition $\matA_{\calI,:}=\zero$ holds when all the independent Bernoulli variables in $\matA_{\calI,:}$ are zeros. The number of independent random variables is $i(N-1)$.
Therefore, we have
\begin{equation}\label{eq:AI=01}
\bbP\lc\matA_{\calI,:}=\zero\rc=(1-p)^{i(N-1)}.
\end{equation}
Further, the entries of $\matA_{\calI^c,\calI^c}\in\bbR^{N-i\times N-i}$ have the same distribution as that of $\matA$. Thus, we apply  \Cref{thm:rank_directed} to get
\begin{equation}\label{eq:AIC1}
\bbP\lc\matA_{\calI^c,\calI^c}\text{ is non-singular}\rc\geq 1-\exp\lb-cp(N-i)\rb,
\end{equation} 
where $c>0$ is universal constant.
Combining \eqref{eq:Ei_relax}, \eqref{eq:AI=01}, and \eqref{eq:AIC1}, we get that
\begin{equation}\label{eq:Ei1}
\bbP\lc\calE_i\rc\geq Q(i,\calU)(1-p)^{i(N-1)} \ls1-\exp\lb-cp(N-i)\rb\rs,
\end{equation}
where $\calE_i$ is defined in \eqref{eq:Ei_defn} and $Q$ is as defined in the statement of the theorem (see \eqref{eq:Q_defn}).

Finally,  we complete the proof by combining \eqref{eq:Ei1} and \eqref{eq:E_relax3} (as we mentioned in the beginning of the proof, \eqref{eq:E_relax3} holds in this case).

\hfill\qed

\emph{Remark: } We note that the bound in \Cref{thm:directed} is not as tight as the result in \Cref{thm:undirected} because of the bound in \eqref{eq:E_relax2} used in the proof of \Cref{thm:undirected,thm:directed} (see \Cref{fig:varyp,fig:varyN}). To be specific, for both cases, we claim that $\calF_1 \supseteq \calF_2$ where
\begin{align*}
\calF_1&\triangleq\lc\exists \calS\in\calU,\calI\subseteq\calS:\bar{\matA}_{\calI^c,:}\text{ is full row rank}\rc\\
\calF_2&\triangleq \Big\{\exists \calS\in\calU,\calI\subseteq\calS:\bar{\matA}_{\calI,:}=\zero\text{ and }\bar{\matA}_{\calI^c,\calI^c} \text{ is non-singular}\Big\}.
\end{align*}
We recall that $\bar{\matA}_{\calI^c,:}\in\bbR^{N-\lv\calI\rv\times N}$ and $\bar{\matA}_{\calI,:}\in\bbR^{\lv\calI\rv\times N}$ are the submatrices of $\bar{\matA}$ formed by rows indexed by $\calI^c=\{1,2,\ldots,N\}\setminus\calI$ and $\calI$, respectively. Also, $\bar{\matA}_{\calI^c,\calI^c}\in\bbR^{N-\lv\calI\rv\times N-\lv\calI\rv}$ is the principal submatrix of $\bar{\matA}$ formed by the rows indexed by $\calI^c$ and the corresponding columns. To understand the difference between the directed and the undirected graph cases, we define another event $\calF_3$ as follows:
\begin{equation*}
\calF_3\triangleq \Big\{\exists \calS\in\calU,\calI\subseteq\calS:\bar{\matA}_{\calI,:}=\zero
\text{ and }\bar{\matA}_{\calI^c,:} \text{ is full row rank}\Big\}.
\end{equation*}
Clearly, $\calF_1\supseteq\calF_3\supseteq\calF_2$. However, for undirected  graphs, $\bar{\matA}$ is a symmetric matrix, and so if $\bar{\matA}_{\calI,:}=\zero$, we have $
\begin{bmatrix}
\bar{\matA}_{\calI,:}\\
\bar{\matA}_{\calI^c,:}
\end{bmatrix} = \begin{bmatrix}
\zero &\zero\\
\zero&\bar{\matA}_{\calI^c,\calI^c}
\end{bmatrix}$.
Hence, when $\bar{\matA}_{\calI,:}=\zero$, we have $\bar{\matA}_{\calI^c,:} = \begin{bmatrix}
\zero&\bar{\matA}_{\calI^c,\calI^c}
\end{bmatrix}$, and therefore, $\bar{\matA}_{\calI^c,:}$ has full row rank if and only if $\bar{\matA}_{\calI^c,\calI^c}$ has full row rank. Further, since $\bar{\matA}_{\calI^c,\calI^c}$ is a square matrix, this is equivalent to $\bar{\matA}_{\calI^c,\calI^c}$ being non-singular. Hence, $\calF_3=\calF_2$ for the undirected graph case. However, for directed graphs, $\calF_3\supset\calF_2$, and thus, the bound is not as tight as the bound for the undirected case.

\hfill\qed

 \section{Proof of \Cref{thm:rankundirected}}\label{app:rank_undirected}
\change{The probability with which a symmetric random matrix with iid, zero mean and unit variance above-diagonal entries (i.e., the entries in the upper triangular portion of a matrix other than the diagonal entries) is invertible is studied in \cite{vershynin2014invertibility}. 
Our result is a generalization of \cite[Theorem 1.5]{vershynin2014invertibility} which is modified to handle the adjacency matrix of an undirected \ER graph with edge probability $p$. Our analysis is based on the concentration of inner product using small ball probabilities whereas t\cite[Theorem 1.5]{vershynin2014invertibility} uses the concentration of quadratic forms using small ball probabilities. We start by introducing some notation and useful results from the literature. }
\subsection{Toolbox}\label{sec:toolbox_rank_undirected}
In this section, we present a concept called small ball probability which  describes the spread of a distribution in space. The results on small ball probabilities requires us to define two other quantities called  L\' evy concentration function and least common denominator (LCD). The definition of the L\' evy concentration function is as follows:
\begin{definition}[L\' evy function]\label{def:levy}
The L\' evy concentration of a random vector $\vecx\in\bbR^N$ for any $\epsilon>0$ is defined as
$\calL(\vecx,\epsilon ) = \underset{\vecz\in\bbR^N}{\sup}\bbP\lc\lV\vecx-\vecz\rV\leq\epsilon \rc$.
\end{definition}
Thus, the  L\'evy concentration function measures the small ball probabilities, namely, the likelihood that the random vector
$\mathbf{x}$  enters a small ball of radius $\epsilon$ in the space. A useful result on L\' evy concentration which we will use to define the LCD is as follows:
\begin{lemma}\label{lem:delta0}
Let $\xi$ be a random variable with unit variance and finite fourth moment, and $\zeta\in\lc0,1\rc$ be another random variable independent of $\xi$ such that $p=\bbP\lc\zeta=1\rc$. Then, there exist constants $0<\delta_0,\epsilon<1$ such that the L\'evy function (in \Cref{def:levy}) satisfies 
$\calL(\zeta\xi,\epsilon)\leq 1-\delta_0 p$.
\end{lemma}
\begin{proof}
The proof follows from~\cite[Lemma 3.2]{rudelson2009smallest} and \cite[ Remark 6.4]{luh2018sparse}.
\end{proof}

We need some other definitions to introduce the concept of LCD. Let $\bbS^{N-1}\subset\bbR^N$ denote the unit Euclidean sphere. We define a subset of $\bbS^{N-1}$ parameterized by $\rho\in(0,1)$ based on sparsity as 
\begin{equation}
\incomp(N,\rho)\triangleq \Big\{\vecx\in\bbS^{N-1}: \nexists\;\vecy\in\bbR^N \text{such that} 
\ld  \lV\vecy\rV_0\leq  \frac{N}{(pN)^{1/16}}, \lV\vecx-\vecy\rV\leq \rho\rc.\label{eq:incomp}
\end{equation}
The set $\incomp(N,\rho)$ represents the set of incompressible vectors, \change{i.e., the} vectors that are not close to sparse vectors with at most $ \frac{N}{(pN)^{1/16}}$ nonzero entries. 
\begin{definition}[Regularized LCD{~\cite[Definition 6.3]{luh2018sparse}}]\label{def:rlcd}
Let $\alpha\in(0,1)$, $\vecx\in\incomp(N,\rho)$ and $\bbZ$ be the set of integers. We define the regularized LCD of $(\vecx,\alpha)$ as 
\begin{equation*}
\hat{D}(\vecx,\alpha)=\underset{\substack{\calI\subset [N]:\lv\calI\rv\leq\left\lceil\alpha N\right\rceil}}{\max} D\lb \vecx_{\calI}/\lV\vecx_{\calI}\rV\rb \text{ with }  D(\vecx) = \inf \lc \theta>0: \mathrm{dist}\lb \theta\vecx,\bbZ^N\rb<\gamma\rc,
\end{equation*}
where $\gamma = (\delta_0p)^{-1/2}\sqrt{\log_+\lb\sqrt{\delta_0p}\theta\rb}$ and $D(\vecx)$ is called the LCD of $\vecx$ and $\delta_0$ is given by \Cref{lem:delta0}.
\end{definition}
Here, $D(\vecx)$ is  the generalization of the least common multiple to real valued numbers. If all the entries of $\vecx$ are rational numbers, then $D(\vecx)$ is the least common multiple of the denominators of the entries of $\vecx$, i.e., $D(x)$ is the smallest integer $\theta$ such that $\theta\vecx\in\bbZ^N$. This quantity $D(\vecx)$ bounds the small ball probabilities of projections, $\vecx^\T\veca$. 
The quantitative relation between $\calL\lb\vecx^\T\veca,\sqrt{p}\epsilon\rb$ and $D(\vecx)$ is provided next. 
\begin{proposition}[{\cite[Proposition 6.5]{luh2018sparse}}]\label{prop:small}
Let $\veca\in\bbR^N$ be a random vector with independent entries $\veca_i=\zeta_i\xi_i$ where $\bbP\lc\zeta_i=1\rc=1-\bbP\lc\zeta_i=0\rc=p$, and $\xi_i$ is a random variable with unit variance and finite fourth moment. Also, $\zeta_i$ and $\xi_i$ are independent random variables. Then, for any $\vecx\in\bbS^{N-1}$  and $\epsilon>0$, the L\'evy function satisfies
$\calL\lb\vecx^\T\veca,\sqrt{p}\epsilon\rb\leq C_1\lb \epsilon +\frac{1}{\sqrt{p}D(\vecx)}\rb$,
where $D$ is the LCD given by \Cref{def:rlcd}.
\end{proposition}

To state the other results used in the proof, we define a subset of $\incomp(N,\rho)$ in \eqref{eq:incomp} based on the regularized LCD (see \Cref{def:rlcd}) as follows:
\begin{equation}
\lar  (N,\rho)\triangleq \Big\{\vecx\in\incomp(N,\rho): 
 \hat{D}\lb \vecx,(pN)^{-1/16}\rb > \exp\lb(pN)^{1/32}\rb\Big\},\label{eq:lar}
\end{equation}
where $\hat{D}$ is the regularized the least common denominator (see \Cref{def:rlcd}). 

The following result shows that, with high probability, the eigenvectors of $\matA$ (defined in \Cref{thm:rankundirected}) belong to the set $\lar(N,\rho)  $.
\begin{lemma}\label{lem:compress}
There exist positive constants $C$ and $c$  such that if  $ N^{-1}<p<1-N^{-1}$, then for any $\lambda\in\bbR$, the following concentration inequality holds: 
\begin{equation*}
\bbP\Bigg\{\exists \vecx\in \bbS^N \setminus\lar(N,\rho) : 
\lV\lb\matA-\lambda\eye\rb\vecx\rV=0\Bigg\}
\leq \exp(-cpN).
\end{equation*}
Here, $\matA\in\bbR^{N\times N}$ and $p$ are defined in \Cref{thm:rankundirected}. Also, $\bbS^{N-1}\subset\bbR^N$ is the unit Euclidean sphere, $\lar(N,\rho)$ is defined in \eqref{eq:lar}, and
$\rho = C^{-\left\lfloor\frac{\log1/(8p)}{\log\sqrt{pN}}\right\rfloor}$.
\end{lemma}
\begin{proof}
See \Cref{sec:compress}.
\end{proof}

The final result of this subsection bounds the infimum of $\lV\matA\vecx\rV$ over incompressible vectors for a general random matrix $\matA$.
\begin{lemma}\label{lem:distance}
Let $\matA\in\bbR^{N\times N}$ be any random matrix with iid columns. Let $\calH\subseteq \bbR^N$ denote the span of all columns of $\matA$ except the first column. Then,  for every $\epsilon >0$, it holds that
\begin{equation*}
\bbP\lc \underset{\vecx\in\incomp(N,\rho)}{\inf}\lV\matA\vecx\rV\leq \frac{\epsilon \rho}{\sqrt{N}}\rc\leq (pN)^{1/16}\bbP\lc\mathrm{dist}\lb\matA_1,\calH\rb\leq \epsilon \rc,
\end{equation*}
where $\incomp(N,\rho)$ is defined in \eqref{eq:incomp}.
\end{lemma}
\begin{proof}
The result is obtained from~\cite[Lemma 3.5]{rudelson2008littlewood} by choosing the first parameter of the compressible set as $(pN)^{-1/16}$ and the fact that columns of $\matA$ are iid.
\end{proof}

Having presented the mathematical tools,  in the next subsection, we formally prove \Cref{thm:rankundirected}.
\subsection{Proof of \Cref{thm:rankundirected}}
We obtain the probability with which $\matA$ is invertible by computing the probability with which the smallest singular value of $\matA$ is positive. Using the union bound and with $\bbS^{N-1}\subset\bbR^{N}$ denoting the unit Euclidean sphere, we have
\begin{equation}
\bbP\lc \matA \text{ is singular}\rc\leq
\bbP\lc\underset{\vecx\in\lar(N,\rho)}{\inf}\lV \matA \vecx\rV = 0\rc+
\bbP\lc\underset{\vecx\in \bbS^{N-1}\setminus \lar(N,\rho)}{\inf}\lV \matA \vecx\rV = 0\rc,\label{eq:boundterms}
\end{equation}
where  $\lar(N,\rho)$ and $\rho$ are given by \eqref{eq:lar} and \eqref{lem:compress}, respectively. In what follows, we upper bound the two terms in \eqref{eq:boundterms}.

Using \Cref{lem:compress}, there exists a constant $c_1>0$ such that
\begin{equation}
\bbP\lc\underset{\vecx\in\lb\bbS^{N-1}\setminus \lar(N,\rho)\rb}{\inf}\lV \matA \vecx\rV = 0\rc\leq \exp(-c_1pN).\label{eq:boundbarA_1}
\end{equation} 
 
Next, we bound the first term in the right hand side of \eqref{eq:boundterms} using \Cref{lem:distance}. To this end, we use \eqref{eq:lar} to get $\lar(N,\rho)\subset\incomp(N,\rho)$  which is defined in \eqref{eq:incomp}. Thus, we deduce that 
\begin{equation*}
\bbP\lc \underset{\vecx\in\lar(N,\rho)}{\inf}\lV\matA\vecx\rV=0\rc\leq\bbP\lc \underset{\vecx\in\incomp(N,\rho)}{\inf}\lV\matA\vecx\rV=0\rc.
\end{equation*}
Further, we write the symmetric matrix
$\matA=\begin{bmatrix}
0 & \veca^\T\in\bbR^{1\times N-1}\\
\veca\in\bbR^{N-1\times 1} &\matA_{\mathrm{sub}}\in\bbR^{N-1\times N-1}
\end{bmatrix}$,
and apply \Cref{lem:distance} to obtain
\begin{equation}
\bbP\lc \underset{\vecx\in\lar(N,\rho)}{\inf}\lV\matA\vecx\rV=0\rc\leq(pN)^{1/16}\bbP\lc\mathrm{dist}\lb\begin{bmatrix} 0 \\\veca\end{bmatrix},\CS{\begin{bmatrix}
\veca^\T\\
\matA_{\mathrm{sub}}
\end{bmatrix}}\rb
=0\rc,\label{eq:largebound}
\end{equation}
where $\CS{\cdot}$ denote the column space of a matrix. The distance term on the right hand side simplifies as follows:
\begin{equation*}
\mathrm{dist}\lb\begin{bmatrix} 0 \\\veca\end{bmatrix},\CS{\begin{bmatrix}
\veca^\T\\
\matA_{\mathrm{sub}}
\end{bmatrix}}\rb \!\geq  \mathrm{dist}\lb\veca,\CS{
\matA_{\mathrm{sub}}}\rb\!=\underset{\vecz\in\CS{
\matA_{\mathrm{sub}}
}}{\min}\lV\veca-\vecz\rV\!= \underset{\substack{\vecz\in\bbS^{N-2}\\\matA_{\mathrm{sub}}\vecz=\zero}}{\max}\vecz^\T\veca.
\end{equation*}
Therefore, from \eqref{eq:largebound}, we have
\begin{align}
\bbP\lc \underset{\vecx\in\lar(N,\rho)  }{\inf}\lV\matA\vecx\rV=0\rc&\leq (pN)^{1/16}\bbP\lc \underset{\substack{\vecz\in\bbS^{N-2}\\\matA_{\mathrm{sub}}\vecz=\zero}}{\max}\vecz^\T\veca \leq 0 \rc\notag\\
&\leq (pN)^{1/16}\bbP\lc \exists \vecz\in\bbS^{N-2}: \matA_{\mathrm{sub}}\vecz=\zero \text{ and } \vecz^\T\veca = 0\rc\notag\\
&\leq
(pN)^{1/16}\bbP\lc\exists \vecz\in\bbS^{N-2}\setminus\!\lar(N-1,\rho')\!\!: \matA_{\mathrm{sub}}\vecz=\zero\rc \notag\\
&\hspace{1.4cm}+(pN)^{1/16}\bbP\lc\exists \vecz\in\lar(N-1,\rho'): \vecz^\T\veca = 0\rc,\label{eq:inter_1}
\end{align}
where $\rho'\triangleq C^{-\left\lfloor\frac{\log1/(8p)}{\log\sqrt{p(N-1)}}\right\rfloor}$ wherein the constant $C$ is same as the constant in \eqref{eq:boundbarA_1}. Next, we use \Cref{lem:compress,lem:delta0} to simplify the two probability terms in \eqref{eq:inter_1}.

Since the entries of $\matA_{\mathrm{sub}}\in\bbR^{N-1\times N-1}$ have the same distribution as that of $\matA$, we again apply \Cref{lem:compress} to get
\begin{equation}\label{eq:splittwo1}
\bbP\lc \exists \vecz\in\bbS^{N-2}\setminus\lar(N-1,\rho'): \matA_{\mathrm{sub}}\vecz=\zero\rc \leq \exp(-c_1p(N-1)), 
\end{equation}

The second term in \eqref{eq:inter_1} can be simplified as follows:
\begin{multline*}
\bbP\lc\exists \vecz\in\lar(N-1,\rho'): \vecz^\T\veca = 0 \rc
\leq\underset{\vecz\in \lar(N-1,\rho')}{\sup}\;\bbP\lc\lv\vecz^\T\veca\rv= 0 \rc \\ \leq\underset{\vecz\in \lar(N-1,\rho')}{\sup}\; \underset{z\in\bbR}{\sup}\;\bbP\lc\lv\vecz^\T\veca-z\rv=0\rc \leq\underset{\vecz\in \lar(N-1,\rho')}{\sup}\; \calL\lb \vecz^\T\veca,0\rb.
\end{multline*}
Further, we note that the entries of $\veca$ have the same distribution as $\zeta\xi$, where $\zeta,\xi\in\lc0,1\rc$ are Bernoulli random variables with probabilities of being 1 as $1/2$ and $2p$, respectively. Thus, \Cref{prop:small} implies that there exists a constant $C_1>0$ with
\begin{equation}
\bbP\lc\exists \vecz\in\lar(N-1,\rho'): \vecz^\T\veca= 0\rc\leq \underset{\vecz\in \lar(N-1,\rho')}{\sup}\;\frac{C_1}{\sqrt{2p} D(\vecz)}\leq \frac{C_1}{\sqrt{2p}\, e^{((N-1)p)^{\frac{1}{32}}}},\label{eq:splittwo2}
\end{equation} 
where the last step follows from the definition of $\lar(N-1,\rho')$ and the fact that $\hat{D}(\vecx,\alpha)\leq D(\vecx)$, for any $\vecx\in\bbS^{N-1}$ and $0<\alpha<1$. Combining \eqref{eq:inter_1},  \eqref{eq:splittwo1} and \eqref{eq:splittwo2}, we get that
\begin{align*}
\bbP\lc \underset{\vecx\in\lar(N,\rho)  }{\inf}\lV\matA\vecx\rV=0\rc& \leq
(pN)^{1/16}\!\exp(-c_1p(N-1))\!+\!\frac{C_1}{\sqrt{2p}\exp\lb((N-1)p)^{1/32}\rb}\\
&\leq C_2\exp\lb-c_2(pN)^{1/32}\rb,
\end{align*}
for some constants $C_2,c_2>0$. Finally, combining the above equation with \eqref{eq:boundterms} and \eqref{eq:boundbarA_1}, we conclude that
$\bbP\lc \underset{\vecx\in\bbS^{N-1}  }{\inf}\lV\matA\vecx\rV> 0\rc
\geq 1- C_3\exp\lb-c_3(pN)^{1/32}\rb$,
for some constants $C_3,c_3>0$. Thus, the proof is complete.

\hfill\qed

\subsection{Proof of \Cref{lem:compress}}\label{sec:compress}
The proof is adapted from \cite[Theorem 2.2]{luh2018sparse}, which relies on the following lemma:
\begin{lemma}\label{lem:connect}
 Let $\calX\subset\bbS^{N-1}\subset\bbR^N$. We fix parameters  $\epsilon>0$, $0\leq\alpha<1/2$, $1/16\leq \beta$ and $\lambda\in\bbR$. Suppose for all $N^{-1}<p\leq 1/2$ and for any $\vecy\in\bbR^N$, there exist constants $C_1,c_1>0$ such that
\begin{equation}\label{eq:condi}
\bbP\lc\exists \vecx\in \calX : 
\lV\ls\matA-p(\one\one^\T-\eye)\rs\vecx-\lambda\vecx\rV\leq \epsilon(pN)^{\alpha}\rc\leq C_1\exp(-c_1(pN)^{\beta}).
\end{equation}
where $\matA\in\bbR^{N\times N}$ and $p$ are defined in \Cref{thm:rankundirected}. Then, there exist constants $C_2,c_2>0$ such that for any  $N^{-1}<p\leq 1-N^{-1}$ and $\lambda\in\bbR$,
\begin{equation}\label{eq:connect}
\bbP\lc\underset{ \vecx\in \calX}{\inf} \; 
\lV\matA\vecx-\lambda\vecx\rV=0\rc
\leq C_2 \exp(-c_2(pN)^{\beta}).
\end{equation}
\end{lemma}
\begin{proof}
We first consider the case where $p<1/2$ and note that for any $\vecx\in\calX$,
\begin{equation*}
p\one\one^\T\vecx=p(\one^\T\vecx)\one\in \calY\triangleq \lc\kappa\one:\kappa\in[-pN,pN]\rc.
\end{equation*}
Then, for any $\lambda\in\bbR$, we have
\begin{equation*}
\underset{\vecy\in\calY}{\inf}\; \underset{\vecx\in \calX}{\inf}\lV\ls\matA-p\one\one^\T+p\eye\rs\vecx-\lambda\vecx-\vecy\rV
\leq \underset{\vecx\in \calX}{\inf} \lV \matA\vecx-(\lambda-p)\vecx\rV.
\end{equation*}
This relation leads to
\begin{equation*}
\bbP\lc \underset{\vecx\in \calX}{\inf} \lV \matA\vecx-\lambda\vecx\rV \leq \epsilon(pN)^\alpha\rc 
\!\leq\!
\bbP\lc \underset{\vecy\in\calY}{\inf}\; \underset{\vecx\in \calX}{\inf}\lV\ls\matA-p\one\one^\T\rs\vecx-\lambda\vecx-\vecy\rV\leq \epsilon(pN)^\alpha\rc\!.
\end{equation*}
Let $\calY_{\mathrm{net}}$ be an $\epsilon(pN)^\alpha-$net of $\calY$ and 
$\lv\calY_{\mathrm{net}}\rv\leq \frac{2pN}{\epsilon(pN)^\alpha}\leq c\exp\lb(pN)^{1/16}\rb$,
for some constant $c>0$.  We then deduce from the triangle inequality that for every $\vecy\in\calY$, there exists $\bar{\vecy}\in\calY_{\mathrm{net}}$ such that
\begin{equation*}
\lV\ls\matA-p\one\one^\T\rs\vecx-\lambda\vecx-\vecy\rV\geq\lV\ls\matA-p\one\one^\T\rs\vecx-\lambda\vecx-\bar{\vecy}\rV-\lV\vecy-\bar{\vecy}\rV.
\end{equation*}
Therefore, taking infimum over $\vecx\in \calX$,$\vecy\in\calY$ and $\bar{\vecy}\in\calY_{\mathrm{net}}$,
\begin{equation*}
\underset{\vecy\in\calY}{\inf}\;\underset{\vecx\in \calX}{\inf}\lV\ls\matA-p\one\one^\T\rs\vecx-\lambda\vecx-\vecy\rV\geq\underset{\vecy\in\calY_{\mathrm{net}}}{\inf} \underset{\vecx\in \calX}{\inf}\lV\ls\matA-p\one\one^\T\rs\vecx-\lambda\vecx-\bar{\vecy}\rV-\epsilon(pN)^\alpha.
\end{equation*}
Consequently, we derive
\begin{equation*}
\bbP\lc \!\underset{\vecx\in \calX}{\inf} \lV \matA\vecx-\lambda\vecx\rV \!\leq\!\epsilon(pN)^\alpha \!\rc 
\!\leq\!
\bbP\lc\! \underset{\vecy\in\calY_{\mathrm{net}}}{\bigcup}\!\underset{\vecx\in \calX}{\inf} \lV\ls\matA-p\one\one^\T\rs\vecx-\lambda\vecx-\vecy\rc\!\leq\! 2\epsilon(pN)^\alpha\!\rc\!.
\end{equation*}
Finally, using the union bound and \eqref{eq:condi}, we arrive at the desired result for $p\leq 1/2$. 

Next, to handle the case where $p>1/2$, we notice that the distribution of $\matA-p(\one\one^\T-\eye)$ is the same as that of $(1-p)(\one\one^\T-\eye)-\tilde{\matA}$ where $\tilde{\matA}$ is the adjacency matrix of an \ER graph with edge probability $1-p\leq 1/2$. Therefore, the distribution of $\lV\ls \matA-p(\one\one^\T-\eye)\rs\vecx-\lambda\vecx\rV$ is the same as that of $\lV\ls \tilde{\matA}-(1-p)(\one\one^\T-\eye)\rs\vecx+\lambda\vecx\rV$. Thus, 
we can use similar arguments as for the case $p\leq 1/2$ to prove that \eqref{eq:connect} holds for all $p>1/2$. Hence, the proof of \Cref{lem:connect} is complete.
\end{proof}

To prove \Cref{lem:compress}, we invoke the union bound to get
\begin{multline}\label{eq:union_bnd}
\bbP\lc\underset{\vecx\in\bbS^{N-1}\setminus\lar(N,\rho)}{\inf}\lV \matA \vecx\rV = 0\rc  \leq 
\bbP\lc\underset{\vecx\in\bbS^{N-1}\setminus\incomp(N,\rho)}{\inf}\lV \matA \vecx\rV = 0\rc \\+
\bbP\lc\underset{\vecx\in\change{\smal(N,\rho)}}{\inf}\lV \matA \vecx\rV = 0\rc,
\end{multline}
where \change{$\smal(N,\rho) \triangleq \change{\smal(N,\rho)} $} we use the fact that $\lar(N,\rho)\subseteq \incomp(N,\rho)$ which in turn, implies that 
$\bbS^{N-1}\setminus\lar(N,\rho)\change{\subseteq}\ls \bbS^{N-1}\setminus\incomp(N,\rho)\rs\bigcup\change{\smal(N,\rho)}$.
 In the following, we show that for each of these two sets, there exist $\epsilon,\beta,\alpha$ satisfying the conditions of \Cref{lem:connect} such that for any given $\lambda\in\bbR$ and $\vecy\in\bbR^N$ and for all $N^{-1}<p<1/2$, \eqref{eq:condi} holds. Thus, using \Cref{lem:connect} with $\lambda=0$ in \eqref{eq:connect} and \eqref{eq:union_bnd}, our proof is complete.

We start by proving that the set $\bbS^{N-1}\setminus\incomp(N,\rho)$ satisfies the condition \eqref{eq:condi} of \Cref{lem:connect}. We use the non-centered version of \cite[Corollary 5.5]{luh2018sparse} given in \cite[Appendix B]{luh2018sparse}. The corollary states that there exist constants $C,c,\tilde{c}>0$ such that if  $2\leq p^{-1}<M$ and for any $\lambda\in \bbR$ and $\vecy\in\bbR^N$,
\begin{align*}
\bbP\Bigg\{\underset{ \vecx\in \comp(N,M,\rho) }{\inf}
\lV\matA-p(\one\one^\T-\eye)\vecx-\vecy\rV\leq \tilde{c}\rho \sqrt{pN}\Bigg\}
&\leq \exp(-cpN)\\
\comp(N,M,\rho)\triangleq \Big\{\vecx\in\bbS^{N-1}: \exists\;\vecy\in\bbR^N \text{such that}
&\ld   \lV\vecy\rV_0\leq  M \text{and} \lV\vecx-\vecy\rV\leq \rho\rc.
\end{align*}
In our case, since $N^{-1}<p<1/2$, we get $p^{-1}<\frac{N}{(pN)^{1/16}}$, and with $M=\frac{N}{(pN)^{1/16}}$, 
\begin{equation*}
\bbP\lc\underset{ \vecx\in \bbS^{N-1}\setminus \incomp(N,\rho)}{\inf}  
\lV\ls\matA-p\one\one^\T-p\eye\rs\vecx-\lambda\vecx\rV\leq \tilde{c}\rho \sqrt{pN}\rc
\leq \exp(-c_1pN),
\end{equation*}
for some constant $c_1>0$ and for any $\vecy\in\bbR$. Therefore, $\bbS^{N-1}\setminus\incomp$ satisfies the condition \eqref{eq:condi} of \Cref{lem:connect}.

Finally, we complete the proof by establishing that $\change{\smal(N,\rho)}$ also satisfies the condition \eqref{eq:condi} of \Cref{lem:connect}. For this, we rely on the following related results:
\begin{enumerate}[label=(\roman*),leftmargin=0.5cm]
\item From the non-centered version of  \cite[Proposition 5.2]{luh2018sparse} (explicitly stated as the equation above Proposition 8.1), we get that  there exist constants $K,c_2$ such that
\begin{equation*}
\bbP\lc\lV\matA-p(\one\one^\T-\eye)\rV>K\sqrt{pN}\rc\leq \exp(c_2pN).
\end{equation*}
\item From \cite[Proposition 6.8, Equation (3)]{luh2018sparse}, for any $\vecx\in\incomp(N,\rho)$, we know that
\begin{equation}\label{eq:setD}
 \hat{D}\lb \vecx,(pN)^{-1/16}\rb\geq \frac{\rho^2\sqrt{\change{N}}}{4(pN)^{3/32}}.
\end{equation}
\item \cite[Proposition 8.1]{luh2018sparse} establishes the following: For any $N^{-1}\leq p\leq 1/2$,  $\lambda\in[-K\sqrt{pN},K\sqrt{pN}]$ and $\vecy\in\bbR^N$, there exist constants $C_3,c_3,\tilde{c}>0$ ,
\begin{equation}\label{eq:mid_low1}
\bbP\Bigg\{\underset{ \vecx\in \hat{\bbS}_D }{\inf}
\lV\matA-p(\one\one^\T-\eye)\vecx-\vecy\rV\leq C_3\epsilon (pN)^{7/16}\Bigg\}
\leq \exp(-c_3pN),
\end{equation}
where $\tilde{c}\frac{\sqrt{\change{N}}}{(pN)^{1/32}}\leq D\leq \exp\lb(pN)^{1/32}\rb$ and we define $\epsilon=\min\lc \frac{\sqrt{N}}{D},\frac{\rho}{4}(pN)^{\frac{1}{8}}\rc$ and
\begin{align}
\hat{\bbS}_D = \lc\vecx\in\incomp\!:\! D\leq\hat{D}\lb \vecx,(pN)^{-\frac{1}{16}}\rb\leq 2D \rc. \label{eq:S_defn}
\end{align}

\item \cite[Proposition 8.2]{luh2018sparse} proves that for any given $\lambda\in[-K\sqrt{pN},K\sqrt{pN}]$, the relation \eqref{eq:mid_low1} holds if $\frac{\rho^2\sqrt{\change{N}}}{4(pN)^{3/32}}\leq D\leq\tilde{c}\frac{\sqrt{\change{N}}}{(pN)^{1/32}}$.
\end{enumerate}
Combining these arguments, we obtain the following result: There exist constants $C_4,c_4>0$ such that if  for any $ N^{-1}<p\leq 1/2$, $\lambda\in\bbR$ and $\vecy\in\bbR^N$,
\begin{equation}\label{eq:mid_low}
\bbP\Bigg\{\underset{ \vecx\in \hat{\bbS}_D }{\inf}
\lV\matA-p(\one\one^\T-\eye)\vecx-\vecy\rV\leq C_4\epsilon (pN)^{7/16}\Bigg\}
\leq \exp(-c_4pN),
\end{equation}
where $\frac{\rho^2\sqrt{\change{N}}}{4(pN)^{3/32}}\leq D\leq \exp\lb(pN)^{1/32}\rb$. Also, from \eqref{eq:setD},  we deduce that
\begin{equation*}
\change{\smal(N,\rho)} \!=\! \lc\! \vecx\in\incomp(N,\rho)\!: \!
\frac{\rho^2\sqrt{n}}{4(pN)^{3/32}}\leq \hat{D}\lb \vecx,(pN)^{-1/16}\rb\!\leq\! \exp\lb(pN)^{1/32}\rb\!\rc\!.
\end{equation*}
Next, we use the covering set-based arguments to prove that $\change{\smal(N,\rho)}$ satisfies Condition \eqref{eq:condi} of \Cref{lem:connect}. We have $\change{\smal(N,\rho)}=\bigcup_{k=1}^K \hat{\bbS}_{2^{-k}\exp\lb(pN)^{1/32}\rb}$ where
\begin{equation*}
K=\min\lc k\in\nat: 2^{-k}\exp\lb(pN)^{1/32}\rb\leq \frac{\rho^2\sqrt{\change{N}}}{4(pN)^{3/32}}\rc\leq (pN)^{1/16},
\end{equation*}
 and $\hat{\bbS}$ is as defined in \eqref{eq:S_defn}. Using \eqref{eq:mid_low} and the union bound, we conclude that
\begin{multline*}
\bbP\lc\!\underset{ \vecx\in \change{\smal(N,\rho)}}{\inf}
\lV\matA-p(\one\one^\T-\eye)\vecx-\vecy\rV\!\leq\! C_4\epsilon (pN)^{\frac{7}{16}}\!\rc
\!\leq\!
(pN)^{\frac{1}{16}}e^{-c_4pN}
\!\leq\! e^{-c_5pN},
\end{multline*}
for some constant $c_5>0$. Thus, the proof is complete.

\hfill\qed

\noindent\textbf{Acknowledgment: } We thank Dr. Rajasekhar Anguluri  from the School of Electrical, Computer and Energy Engineering at Arizona State University for suggesting a shorter proof for \Cref{lem:Hadamard}.
\bibliographystyle{siamplain}

\end{document}